\documentclass[letterpaper,UKenglish,cleveref, autoref, thm-restate]{lipics-v2021}

\hideLIPIcs  

\title{New Algorithm for Combinatorial n-folds and Applications}

\titlerunning{Combinatorial n-folds and Applications} 

\usepackage{marvosym}
\usepackage{todonotes}

\usepackage{csquotes}

\author{Klaus Jansen}{Kiel University, Department of Computer Science,  Germany}{kj@informatik.uni-kiel.de}{}{}
\author{Kai Kahler}{Kiel University, Department of Computer Science,  Germany}{kka@informatik.uni-kiel.de}{}{}
\author{Lis Pirotton}{Kiel University, Department of Computer Science,  Germany}{lpi@informatik.uni-kiel.de}{}{}
\author{Malte Tutas}{Kiel University, Department of Computer Science,  Germany}{mtu@informatik.uni-kiel.de}{}{}

\authorrunning{K. Jansen, K. Kahler, L. Pirotton and M. Tutas} 

\Copyright{K. Jansen, K. Kahler, L. Pirotton and M. Tutas} 



\category{} 

\relatedversion{} 

\funding{Funded by the Deutsche Forschungsgemeinschaft (DFG, German Research Foundation) – Project number 528381760"}

\acknowledgements{We thank Martin Koutecký for the helpful discussions regarding the topic of $P\|C_{\max}$.}

\EventEditors{John Q. Open and Joan R. Access}
\EventNoEds{2}
\EventLongTitle{42nd Conference on Very Important Topics (CVIT 2016)}
\EventShortTitle{CVIT 2016}
\EventAcronym{CVIT}
\EventYear{2016}
\EventDate{December 24--27, 2016}
\EventLocation{Little Whinging, United Kingdom}
\EventLogo{}
\SeriesVolume{42}
\ArticleNo{23}

\nolinenumbers
\usepackage{xspace}

\usepackage{algorithm}
\usepackage[noend]{algpseudocode}
\makeatletter
\AfterEndEnvironment{algorithm}{\let\@algcomment\relax}
\AtEndEnvironment{algorithm}{\kern2pt\hrule\relax\vskip3pt\@algcomment}
\let\@algcomment\relax
\newcommand\algcomment[1]{\def\@algcomment{\footnotesize#1}}

\renewcommand\fs@ruled{\def\@fs@cfont{\bfseries}\let\@fs@capt\floatc@ruled
  \def\@fs@pre{\hrule height.8pt depth0pt \kern2pt}%
  \def\@fs@post{}%
  \def\@fs@mid{\kern2pt\hrule\kern2pt}%
  \let\@fs@iftopcapt\iftrue}
\makeatother

\usepackage{tikz}
\usetikzlibrary{shapes.geometric, arrows.meta}
\usetikzlibrary{graphs}
\usetikzlibrary{matrix}
\usetikzlibrary{decorations.pathreplacing}
\tikzset{%
	>={Latex[width=2mm,length=2mm]},
	base/.style = {rectangle, rounded corners, draw=black,
		minimum width=4cm, minimum height=1cm,
		text centered, font=\sffamily},
	activityStarts/.style = {base, fill=blue!30},
	startstop/.style = {base, fill=red!30},
	activityRuns/.style = {base, fill=subbi!30},
	process/.style = {base, minimum width=2.5cm, fill=orange!15,
		font=\ttfamily},
}
\tikzset{toprule/.style={%
		execute at end cell={%
			\draw [line cap=rect,#1] (\tikzmatrixname-\the\pgfmatrixcurrentrow-\the\pgfmatrixcurrentcolumn.north west) -- (\tikzmatrixname-\the\pgfmatrixcurrentrow-\the\pgfmatrixcurrentcolumn.north east);%
		}
	},
	bottomrule/.style={%
		execute at end cell={%
			\draw [line cap=rect,#1] (\tikzmatrixname-\the\pgfmatrixcurrentrow-\the\pgfmatrixcurrentcolumn.south west) -- (\tikzmatrixname-\the\pgfmatrixcurrentrow-\the\pgfmatrixcurrentcolumn.south east);%
		}
	}
}
\usetikzlibrary{patterns}
\usetikzlibrary{decorations,arrows}
\usetikzlibrary{decorations.pathmorphing}
\usepgflibrary{decorations.pathreplacing} 
\usetikzlibrary{decorations.text}

\usepackage{nicefrac}
\usepackage{xfrac}

\usepackage{dsfont}

\newcommand*{\bdown}[1][]{b^{\downarrow #1}}
\newcommand*{\bup}[1][]{b^{\uparrow #1}}

\newcommand{\supp}{\textup{supp}}
\newcommand{\maxSupp}{K}
\newcommand{\maxSuppFormula}{\lfloor 2 \cdot (r + 1)\cdot \log(4 \cdot (r + 1) \cdot \Delta) \rfloor}

\newcommand{\txtotherwise}{\textup{otherwise}}

\newcommand{\calI}{\mathcal{I}}
\newcommand{\calIFormula}{
        \begin{cases}
            \log \big(\nicefrac{(\bdown_{\max} + \maxSupp)}{(2\maxSupp+1)} \big) +2, &\txtif \log \big(\nicefrac{(\bdown_{\max} + \maxSupp)}{(2\maxSupp+1)} \big) \textup{ is integer}\\
            \big\lceil\log \big(\nicefrac{(\bdown_{\max} + \maxSupp)}{(2\maxSupp+1)} \big)\big\rceil +1, &\txtotherwise.
        \end{cases}}
\newcommand{\calIkFormula}{
        \begin{cases}
            \log \big(\nicefrac{(\bdown_k + \maxSupp)}{(2\maxSupp+1)} \big) +2, &\txtif \log \big(\nicefrac{(\bdown_k + \maxSupp)}{(2\maxSupp+1)} \big) \textup{ is integer}\\
            \big\lceil\log \big(\nicefrac{(\bdown_k + \maxSupp)}{(2\maxSupp+1)} \big)\big\rceil +1, &\txtotherwise.
        \end{cases}}
\newcommand{\calO}{O}

\newcommand{\calA}{\mathcal{A}}
\newcommand{\calC}{\mathcal{C}}
\newcommand{\calN}{\mathcal{N}}

\newcommand{\pmaxD}{\pmax^\Od}
\newcommand{\pmaxOne}{\pmax^\Oone}
\newcommand{\Od}{{\calO(d)}}
\newcommand{\Or}{{\calO(r)}}
\newcommand{\Oone}{{\calO(1)}}
\newcommand{\trueCapital}{\texttt{true}}
\newcommand{\falseCapital}{\texttt{false}}

\newcommand{\polyI}{|I|^{\Oh(1)}}

\newcommand{\Z}{\mathbb{Z}}

\newcommand{\ZZ}{\mathbb{Z}}
\newcommand{\ZZgeqzero}{\mathbb{Z}_{\geq 0}}
\newcommand{\ZZgeqone}{\mathbb{Z}_{\geq 1}}
\newcommand{\ii}{(i)}
\newcommand{\powIteri}{{(i)}}

\newcommand{\powIterOne}{{(1)}}
\newcommand{\powIterTwo}{{(2)}}
\newcommand{\powIterCalI}{{(\calI)}}
\newcommand{\powIteriminOne}{{(i-1)}}
\newcommand{\powIteriplusOne}{{(i+1)}}

\newcommand{\Cmax}{C_{\max}}
\newcommand{\Cmin}{C_{\min}}

\newcommand{\QCmax}{Q\|C_{\max}}
\newcommand{\QCmin}{Q\|C_{\min}}

\newcommand{\tiln}[1][]{\tilde{b}^{\uparrow #1}}
\newcommand{\tilm}[1][]{\tilde{b}^{\downarrow #1}}
\newcommand{\tilx}{\tilde{x}}

\newcommand{\hatn}[1][]{\hat{b}^{\uparrow #1}}
\newcommand{\hatm}[1][]{\hat{b}^{\downarrow #1}}
\newcommand{\hatx}{\hat{x}}

\newcommand{\RTfeasibilityO}{(n r \Delta)^\Or \log(b_{\text{def}})}
\newcommand{\RTsched}{(d \pmax)^\Od \log(m_{\text{def}})}
\newcommand{\RTqcmaxO}{\pmax^\Od \polyI}

\newcommand{\closeststring}{\textsc{Closest String}\xspace}
\newcommand{\imbalance}{\textsc{Imbalance}\xspace}

\newcommand{\txtif}{\textup{if }}
\newcommand{\txtother}{\textup{otherwise }}

\newcommand{\pmax}{p_{\max}}

\newcommand{\Oh}{O}
\newcommand{\etal}{et al.\xspace}
\newcommand{\norm}[1]{\left\lVert#1\right\rVert}

\newcommand{\nfold}{$n$-fold\xspace}

\newcommand{\nullvec}{\mathds{O}}
\newcommand{\onevec}{\mathds{1}}

\begin{document}
\maketitle

\begin{abstract}
 Block-structured integer linear programs (ILPs) play an important role in various application fields. We address $n$-fold ILPs where the matrix $\mathcal{A}$ has a specific structure, i.e., where the blocks in the lower part of $\mathcal{A}$ consist only of the row vectors $(1,\dots,1)$.
In this paper, we propose an approach tailored to exactly these combinatorial $n$-folds. We utilize a divide and conquer approach to separate the original problem such that the right-hand side iteratively decreases in size. We show that this decrease in size can be calculated such that we only need to consider a bounded amount of possible right-hand sides. This, in turn, lets us efficiently combine solutions of the smaller right-hand sides to solve the original problem. We can decide the feasibility of, and also optimally solve, such problems in time $(n r \Delta)^{O(r)} \log(\|b\|_\infty),$ where $n$ is the number of blocks, $r$ the number of rows in the upper blocks and $\Delta=\|A\|_\infty$.

We complement the algorithm by discussing applications of the $n$-fold ILPs with the specific structure we require. We consider the problems of (i) scheduling on uniform machines, (ii) closest string and (iii) (graph) imbalance.
Regarding (i), our algorithm results in running times of $p_{\max}^{O(d)}|I|^{O(1)},$ matching a lower bound derived via ETH. 
For (ii) we achieve running times matching the current state-of-the-art in the general case. In contrast to the state-of-the-art, our result can leverage a bounded number of column-types to yield an improved running time. 
For (iii), we improve the parameter dependency on the size of the vertex cover.
\end{abstract}

\newpage
\section{Introduction}\label{sec:introduction}
At an abstract level, computer science is about solving hard problems. In a seminal work, Karp formulated a list of twenty-one NP-hard problems~\cite{Karp72}. Among these is the problem of integer linear programming.
An integer linear program (ILP) in standard form is defined by 
$$\max\{c^Tx \mid \calA x=b, x \in \ZZgeqzero\},$$
with constraint matrix $\calA \in \Z^{m\times n}$, right-hand side (RHS) $b \in \Z^m$ and objective function vector $c \in \Z^n.$
The constraints $Ax=b$ and the objective function $c^T x$ are required to be linear functions. 
Such integer programs have been the subject of vast amounts of research, such as~\cite{Papadimitriou81,Lenstra83,ReisR23,JR23}, with the latter two results representing the current state of the art of running times regarding $n$ and $m$ respectively. However, with integer programming being an NP-hard problem, research has flourished into different structures of integer programs, which utilize a more specific structure to yield improved running times. One such structure is the \nfold ILP.  

\paragraph*{\nfold ILP}
Let $r, s, t, n \in \ZZgeqzero$. An \nfold ILP is an integer linear program with RHS vector $b \in \ZZ^{r+ns}$ and with a constraint matrix of the form
$$
\begin{pmatrix}
    A_1     & A_2   & \dots  & A_n      \\
    B_1     & 0     & \dots  & 0        \\
    0       & B_2   &        & \vdots   \\
    \vdots  &       & \ddots & 0        \\
    0       & \dots & 0      & B_n
\end{pmatrix},
$$
where $A_1, \dots, A_n \in \ZZ^{r \times t}$ and $B_1, \dots, B_n \in \ZZ^{s \times t}$ are matrices, and $x = (x_1, \dots, x_n)$ is the desired solution. In other words, if we delete the first $r$ constraints, 

which are called the \emph{global} constraints, the problem decomposes into independent programs, 

which are called \emph{local}. If $\Delta$ is the largest absolute value in $A_i$ and $B_i$, \nfold integer programs can be solved in time $(nt)^{(1+o(1))} \cdot 2^{O(rs^2)}(rs\Delta)^{O(r^2s+s^2)}$~\cite{CEHRW21}. This algorithm comes as the latest development in a series of improvements spanning almost twenty years~\cite{AschenbrennerH07,CEHRW21,CslovjecsekEPVW21,DeLoera2008,EHK18,HemmeckeOR13,JansenLR20,KouteckyLO18}. \nfold integer programs have already found applications in countless problems, often yielding improved running times over other techniques. As the total number of applications is too large, we refer the reader to~\cite{ChenMYZ17,DeLoera2008,GavenciakKK22,JansenKMR22,KnopK18,KnopKLMO21,KnopKM20} for some examples. 

In the following, we focus on combinatorial \nfold ILPs where we allow blocks of different width, i.e., $t \in \ZZgeqzero^n$ is a vector and therefore, the block $i \in [n]$ has $t_i$ columns. Also, we have $B_i = (1,\dots,1)^T = \mathds{1}_{t_i}$ for all $i \in [n]$ which implies $s = 1$. Thus, we have the following ILP
\begin{align} \label{eq:nfold} 
    \calA x = \begin{pmatrix}
    A_1     & A_2   & \dots  & A_n      \\
    \mathds{1}_{t_1}& 0      & \dots  & 0        \\
    0       & \mathds{1}_{t_2}   &        & \vdots   \\
    \vdots  &       & \ddots & 0        \\
    0       & \dots & 0      & \mathds{1}_{t_n}
\end{pmatrix} x = b, \qquad
x \in \ZZgeqzero^h
\tag{$\spadesuit$}
\end{align}
with $A_i \in \ZZ^{r \times t_i}$ for all $i \in [n]$, $b \in \ZZ^{r+n}$ and where $h := \sum_{i \in [n]}t_i$ is the total number of columns in $\calA$. 
The currently best known algorithms for such combinatorial \nfold ILPs have parameter dependency $(r\Delta)^{O(r^2)}$~\cite{CEHRW21,KnopKM20}.
In our work, we often consider the upper part of the vector $b$ differently from its lower part. Therefore, we introduce the following notation: $b = (\bup, \bdown)^T$ with $\bup \in \ZZ^r$ and $\bdown \in \ZZ^n$. Also, define $\bup_{\max} := \max_{j\in[r]} \bup_j,$ $\bdown_{\max} := \max_{k\in[n]} \bdown_k$ and $b_{\text{def}} := \min(\bup_{\max},\bdown_{\max})$.

\paragraph*{Our contribution}
In this paper, we provide a novel approach to solving such combinatorial \nfold ILPs. While the idea of the algorithm we provide is inspired by Jansen and Rohwedder~\cite{JR23}, we note that we utilize the structure of \nfold ILPs to generate different techniques to split solutions apart into smaller sub-solutions. Using these techniques, we achieve the following result.

\begin{restatable}{theorem}{thmopt}\label{thm:opt}
    Combinatorial \nfold ILPs \eqref{eq:nfold} with respect to the objective of maximizing $c^Tx, c \in \ZZgeqzero^h$ and where $\Delta$ is the largest absolute value in $\calA$ can be solved in time $\RTfeasibilityO$.
\end{restatable}
Comparatively, the algorithm in~\cite{JR23} yields a running time of 
$O(\sqrt{r+n}\Delta)^{2(r+n)} + O(h\cdot(r+n))$ when applied to our setting, while the state-of-the-art algorithm~\cite{CEHRW21} achieves a running time of $(n\|t\|_\infty)^{(1+o(1))} \cdot 2^{O(r)}(r\Delta)^{O(r^2)}$.

To complement the general techniques we provide to solve combinatorial \nfold ILPs, we show how our result can be applied to common problems. First, we consider scheduling on uniform machines with the objective of makespan minimization $(\QCmax)$ and Santa Claus $(\QCmin)$. 
After publication of a preliminary version of our results, Rohwedder~\cite{rohwedder2025} recently bounded $\Delta$ by $\pmaxOne$ for $\QCmax$, where $\pmax$ is the largest processing time.
They also provide an (almost) tight algorithm for $\QCmax$ with running time $(d\pmax)^{O(d)}(h+M)^{O(1)} |I|^{O(1)}$, where $d$ is the number of job-types, $h$ is the total number of columns in the configuration ILP and $M$ is the total number of machines. Replacing a subroutine in \cite{rohwedder2025} by our algorithm improves their result to $\RTsched |I|^{O(1)}$, where $m_{\text{def}}:=\min(n_{\max},m_{\max})$ i.e., we achieve a better result with respect to parameter $h$ and $M$. This answers the open question by Kouteck\'y and Zink in~\cite{KZ20}. They pose the question whether $\QCmax$ can be solved in time $\pmaxD \polyI$. 
In more detail, the calculation of the Multi-Choice IP in~\cite{rohwedder2025} has a running time of $(nr\Delta)^{O(r)}(h+\|\bdown\|_1)$, while our \nfold algorithm runs in $\RTfeasibilityO$, i.e., solves these Multi-Choice IPs more quickly.

We also manage to extend these techniques to $\QCmin.$ This yields running times that are (almost) tight by an ETH-lower bound.

\begin{restatable}{theorem}{schedthm}\label{thm:main2}
    $\QCmax$ and $\QCmin$ can be solved in time $\RTqcmaxO$.
\end{restatable}

We also apply our algorithm to the \closeststring problem and the \imbalance problem (see \cref{sec:Applications} for formal definitions). 
The current best running times~\cite{KnopKM20} for both problems have a quadratic dependency of $k$ in the exponent, where $k$ is the number of considered strings (\closeststring), respectively the size of the vertex cover (\imbalance). When the number of column-types (\closeststring), respectively vertex-types (\imbalance) is bounded by $T$, our algorithm achieves the following results.
\begin{restatable}{corollary}{closestStringCor}\label{cor:closest string}
    The \closeststring problem can be solved in time $((T+1) k)^{O(k)} \log(L)$.
\end{restatable}
\begin{restatable}{corollary}{imbalanceProb}\label{cor:imbalance}
    The \imbalance problem with $n$ vertices can be solved in time $((T+2)k)^{O(k)} \log(kn)$, where $k$ is the size of a vertex cover and $T$ is the number of vertex-types.
\end{restatable}
Generally, $T$ is bounded by $k^k$ for \closeststring. In this case, we meet the state-of-the-art parameter dependency. For \imbalance, $T$ may take any integer value up to $2^k$. The resulting parameter dependency of $2^{k^2}$ improves on the state-of-the-art $k^{k^2}$ which is the application of~\cite{KnopKM20} to the problem.
Additionally, when $T$ is bounded by $k^{O(1)}$, we are able to reduce the dependency in the exponent to linear for both problems.
When applying Rohwedder's algorithm for the Multi-Choice IP~\cite[Theorem 2]{rohwedder2025} to those problems, the resulting running times are $((T+1)k)^{O(k)}O(T+dk)$, where $d$ is the upper bound on string distance (\closeststring) and 
$((T+2)k)^{O(k)}O(Tk+nk)$ (\imbalance).

We expect that our techniques can be extended to several other problems that can be formulated as a combinatorial \nfold ILP as well, such as the string problems given in~\cite{KnopKM20}.

\subsection{Related Work}
Integer linear programming is one of the twenty-one problems originally proven to be NP-hard by Karp in his seminal work~\cite{Karp72}. Since then, ILPs have been used to solve countless NP-hard problems, for examples see \mbox{e.g.\ \cite{GavenciakKK22}.} The currently best known running times can be separated by their parameterization. When considering a small number of constraints $m$, the best known running time is given by Jansen and Rohwedder~\cite{JR23}, and runs in time $O(\sqrt{m}\Delta)^{(1+o(1))m}+O(nm).$ For a small number of variables $n$, the current best known algorithm is given by Reis and Rothvoss~\cite{ReisR23}, with a running time of $(\log n)^{O(n)}\cdot |\mathcal{A}|^{O(1)}$.

With the general ILP being a hard problem, research into certain substructures that are more easily solvable began. For an overview of this process, see e.g.,~\cite{EHKKLO22,GavenciakKK22}. Of particular interest to this paper is the \nfold ILP, whose current best known algorithm is given by Cslovjecsek, Eisenbrand, Hunkenschröder, Rohwedder and Weismantel~\cite{CEHRW21}. Their algorithm runs in time $(nt)^{(1+o(1))} \cdot 2^{O(rs^2)}(rs\Delta)^{O(r^2s+s^2)}$ and is the latest result in a series of improvements that spans close to twenty years, see~\cite{CslovjecsekEPVW21,DeLoera2008,EHK18,HemmeckeOR13,JansenLR20,KouteckyLO18} for works complementing this line of research.

Scheduling is one of the most well-known operations research problems. We now specifically turn towards uniform machines. By formulating the problem as an \nfold ILP, Knop and Koutecký~\cite{KK18} showed that $\QCmax$ can be solved in time $\pmax^{\Oh(\pmax^2)}\polyI$ if the number of machines is encoded in unary. 
By solving a relaxation to schedule many of the jobs in advance and then applying a dynamic program, Koutecký and Zink~\cite{KZ20} improved this to $\pmax^{\Oh(d^2)}\polyI$, still with the number of machines $M$ encoded in unary. Leveraging the huge-\nfold machinery, Knop \etal~\cite{KKLMO23} get rid of the assumption on $M$ and obtain a $\pmax^{\Oh(d^2)}\polyI$-time algorithm. Kouteck\'y and Zink~\cite{KZ20} pose the open question whether the dependency $\pmax^{\Oh(d^2)}$ for the objective $\QCmax$ can be improved, as it has been for identical machines. In a recent result, Rohwedder~\cite{rohwedder2025} gives an algorithm for $\QCmax$ running in time $\pmax^{O(d)}\polyI.$ We note that our work has been compiled independently and prior to Rohwedders result, but we can leverage his novel ideas to reduce the value of $\Delta$ to guarantee a similar result. Chen, Jansen and Zhang show that, for a given makespan $T$, there is no algorithm solving $P||\Cmax$ in time $2^{T^{1-\delta}}\polyI$~\cite{ChenJZ18}. Jansen, Kahler and Zwanger extended these results to derive a lower bound of $\pmax^{O(d^{1-\delta})}\polyI$~\cite{JansenKZ25}. As $\QCmax$ is a generalization of $P||\Cmax$, these results extend to it.

Regarding the \closeststring problem with $k$ input strings of length $L$, Gramm \etal~\cite{GNR03} formulated the problem as an ILP of dimension $k^{O(k)}$. This led to an algorithm with running time $2^{2^{O(k\log (k))}} O(\log(L))$. Knop \etal~\cite{KnopKM20} improved this double-exponential dependency to the currently best known algorithm with running time $k^{O(k^2)} O(\log(L))$. Rohwedder and W\k{e}grzycki~\cite{RohwedderW25} showed that the \closeststring problem (with binary alphabet) cannot be solved in time $2^{O(k^{2-\varepsilon})}\text{poly}(L)$ with $\varepsilon > 0$ unless the ILP Hypothesis fails.
The ILP Hypothesis states that for every $\varepsilon>0$, there is no $2^{O(m^{2-\varepsilon})}\text{poly}(L)$-time algorithm for ILPs with $\Delta = O(1)$.

Imbalance is known to be NP-complete for various special graph classes~\cite{BiedlCGHW05,KaraKW07}. 
Fellows, Lokshtanov, Misra, Rosamond and Saurabh~\cite{FellowsLMRS08} showed that the problem is fixed-parameter tractable (FPT) parameterized by the size $k$ of the vertex cover. Their algorithm has a double exponential running time. 
There also exist FPT algorithms for this problem parameterized by other parameters, for instance imbalance~\cite{LokshtanovMS13}, neighborhood diversity~\cite{Bakken18} and the combined parameter of treewidth and maximum degree~\cite{LokshtanovMS13}.
Misra and Mittal~\cite{MisraM21} showed that \imbalance is in XP when parameterized by twin cover, and FPT when parameterized by the twin cover and the size of the largest clique outside the twin cover.

\section{Overview}
In this section, we give an overview of the paper and the techniques within. We begin by giving some concepts essential to the functionality of our algorithm. Then, we prove the feasibility version of \cref{thm:opt} under the assumption that $\calA$ contains only non-negative entries. At its core, we use a divide and conquer approach to separate the entire problem into smaller sub-problems, which we can then solve more quickly via dynamic programming. Then, we combine these small solutions such that they solve iteratively larger sub-problems, again requiring a novel structural result to show that such combinations are feasible. A core strength of this dynamic programming formulation and approach is that we only need to solve the dynamic program once for each recombination step. This procedure can be iterated until we have solved the entire \nfold ILP. After that, we show how we can reduce the general problem which also allows negative coefficients to the special case where $\calA$ is non-negative. After that, we prove \cref{thm:opt}. 

The algorithm we provide is complemented by a discussion on its applications to relevant problems. We begin this by elaborating on the problem of scheduling on uniform machines in \cref{sec:Applications}. Here, we can adapt some novel concepts to our techniques such that we achieve a running time that is (almost) tight according to a lower bound derived via the exponential time hypothesis. This yields \cref{thm:main2}. Finally, in \cref{sec:closest string} and \cref{sec:imbalance} we provide a brief outlook onto other relevant problems, the \closeststring problem which results in \cref{cor:closest string} and the \imbalance problem which results in \cref{cor:imbalance}.

Due to space constraints, some proofs are placed in the appendix. We denote the corresponding statements with (\Rightscissors).

\subsection{Preliminaries}
Before we define concepts essential to the algorithm and its functionality, we first specify some notation.
We denote $[k] := \{i \in\ZZgeqzero \mid 1 \le i \le k\}$ and use base 2 logarithms, i.e., we assume $\log(2) = 1.$
A vector that contains a certain value in all entries is stylized, e.g. the zero vector $0_d$ is written as $\nullvec_d.$
We omit the dimension $d$ if it is apparent from the context.
For vector $v,$ we refer to its maximum value $\|v\|_\infty$ with $v_{\max}.$
The support of a $d$-dimensional vector $v$ is the set of its indices with a non-zero entry. It is denoted by $\supp(v) = \{i \in [d] \mid v_i \\ \neq 0\}.$
We define a \textit{brick} $x^{(i)} \in \mathbb{Z}^{t_i}$ as the $i$-th block of the solution vector of \eqref{eq:nfold}, i.e., we have $x = \begin{pmatrix}
    x^{(1)} 
    x^{(2)} 
    \dots 
    x^{(n)}
    \end{pmatrix}^T,$
with $x^{(i)} \in \mathbb{Z}^{t_i}$ for all $i \in [n].$

\paragraph*{Support of a Solution}
By applying the result by Eisenbrand and Shmonin~\cite[Theorem 1 (ii)]{ES06} to the \nfold ILP, one can show the following support bound:

\begin{lemma}\label{lem:support}
    Consider the ILP \eqref{eq:nfold} and let $\Delta$ be the largest absolute value of any entry in the coefficient matrix $\calA$.
    Then there exists a solution $x= (x_1,\dots,x_n)^T$ to \eqref{eq:nfold} with 
    \begin{align*}
        |\supp(x_k)| \leq 2 (r+1) \log(4 (r+1) \Delta), \qquad \forall k\in[n].
    \end{align*}
\end{lemma}
\begin{proof}
    Let $y= (y_1,\dots,y_n)^T$ be a solution to \eqref{eq:nfold}. Then for each $k\in[n]$, the ILP
    \begin{align*}
        \begin{pmatrix}
            A_k \\
            \onevec^T
        \end{pmatrix}
        x_k &=
        \begin{pmatrix}
            A_ky_k \\
            b_{r+k}
        \end{pmatrix} \quad
        x_k\in\ZZgeqzero^c
    \end{align*}
    is feasible and we can apply the support bound by Eisenbrand and Shmonin~\cite[Theorem 1 (ii)]{ES06}. Hence, there exists a solution $x_k$ with $|\supp(x_k)| \leq 2 (r+1) \log(4 (r+1) \Delta).$ Combining these solutions for all $k\in[n]$ yields a solution $x$ with the desired property.
\end{proof}

Note that this bound only holds for the feasibility problem. Using the support bound by Eisenbrand and Shmonin~\cite[Corollary 5 (ii)]{ES06}, the same technique achieves the following result $|\supp(x_k)| \leq 2 (r+2) (\log(r+2) + \Delta+2)$. Other bounds, for instance~\cite{BJMS23}, can also be applied in our approach and may achieve better results in specific application settings.

\section[Combinatorial n-fold Algorithm]{Combinatorial \nfold Algorithm}\label{sec:algorithm}

\begin{figure}
    \centering
\begin{tikzpicture}[scale=0.75]

\filldraw (0.75,0.75) circle (2pt) node[anchor=south]{$0$};
\filldraw (3.0625,1.875) circle (2pt) node[anchor=south]{$\nicefrac{\bup}{4}$};
\filldraw (5.375,3) circle (2pt) node[anchor=south]{$\nicefrac{\bup}{2}$};
\filldraw (10,5.25) circle (2pt) node[anchor=south]{$\bup$};

\draw (0,0) rectangle ++(1.5,1.5);
\draw (2.3125,1.125) rectangle ++(1.5,1.5);
\draw (4.625,2.25) rectangle ++(1.5,1.5);
\draw (9.25,4.5) rectangle ++(1.5,1.5);

\draw[->,blue, thick] (0.75,0.75) -- (1.4,1.075);
\draw[->,red,dashed,thick] (1.4,1.075) -- (2.05,1.4);
\draw[->,blue,thick] (2.05,1.4) -- (2.75,1.45);
\draw[->,red,dashed,thick] (2.75,1.45) -- (4.75,2.15);
\draw[->,blue,thick] (4.75,2.15) -- (5.1,2.75);
\draw[->,red,dashed,thick] (5.1,2.75) -- (9.45,4.75);
\draw[->,blue,thick] (9.45,4.75) -- (10,5.25);
\end{tikzpicture}
   \caption{Example of the procedure with four iterations. The continuous (blue) arrows indicate a feasible upper RHS of the sub-problem in the current iteration. The dashed (red) arrows double the intermediate solution of the previous iteration.}
    \label{fig:boxes}
\end{figure}

On a high level, our algorithm implements an elegant idea. Instead of solving the -- often times large -- entire \nfold ILP in one fell swoop, we manage to reduce it to a set of smaller problems, which we then use to solve the large problem. More specifically, we show that there exists a feasible solution to \eqref{eq:nfold} that can be constructed by doubling a solution {\em near} $\nicefrac{\bup}{2}$ and adding the solution of a small sub-problem. This concept can be repeated iteratively, i.e., a solution near $\nicefrac{\bup}{2}$ can be computed by doubling a solution near $\nicefrac{\bup}{4}$ and adding the solution of another small sub-problem, and so on. We repeat this until the remaining problem only contains a small upper RHS $\bup$. For an illustration of this process, see \cref{fig:boxes}. A core strength of this approach is that we only need to compute this dynamic program once for every iteration, independent of the amount of candidate solutions. 

Taking a closer look, we use insights generated from the structure of our repeatedly doubled intermediary solutions to uniquely determine the lower RHS $\bdown$ at every step of the algorithm. As we double solutions in each step, we can calculate the total number of iterations our algorithm needs. We denote this number by $\calI \in \ZZgeqone$.
Later in this section, we give the exact value of $\calI$ and provide a proof that our algorithm needs exactly $\calI$ iterations.

For any iteration $i\in [\calI]$, we denote the lower RHS by $\bdown[\ii]$. In the first iteration, and every time we solve a small sub-problem, we know by \cref{lem:support} that the support of the individual bricks is bounded by
\begin{align} \label{eq:maxSupp}
    \maxSupp := \maxSuppFormula,
\end{align} 
which also limits the upper RHS, which we refer to as $D$, of these sub-problems. Later, we show how this bound $D$ is determined. The more detailed process is illustrated in \cref{fig:oneIter}. Here, the continuous (blue) arrow contains the solution to the small sub-problem, while the dashed (red) arrow represents the doubling of an earlier solution.

We begin the description of our algorithm by showing that  the lower RHS in each of these iterations is uniquely determinable. We further show that we can indeed find solutions close to each $\nicefrac{\bup}{2^{\calI-i}}$ that can be doubled to produce a feasible solution. Finally, we show how to utilize these insights to construct a feasible schedule while only needing to iteratively solve small sub-problems.

\begin{figure}
    \centering
\begin{tikzpicture}[scale=0.75]
\coordinate (n) at (0,0);
\filldraw[gray] (1,1) circle (2pt) node[anchor=south]{$\nicefrac{\bup}{2^{\calI-(i-1)}}$};
\filldraw[gray] (6,3) circle (2pt) node[anchor=south]{$\nicefrac{\bup}{2^{\calI-i}}$};

\draw[gray] (n) rectangle ++(2,2);

\coordinate (2n) at (5,2);

\draw[gray] (2n) rectangle ++(2,2);

\filldraw (0.5,0.2) circle (1pt);
\filldraw (4.4,2.1) circle (1pt);
\filldraw (5.6,2.3) circle (1pt);

\draw[decorate,decoration={brace,amplitude=10pt, mirror, raise=4pt}] (7,3) -- (7,4) node[midway, right=12pt] {$D$};
\draw[decorate,decoration={brace,amplitude=10pt, raise=4pt}] (0,1) -- (0,2) node[midway, left=12pt] {$D$};
\draw[decorate,decoration={brace,amplitude=10pt,mirror, raise=4pt}] (4.4,1.8) -- (5.6,1.8) node[midway, below=12pt] {$\le D$};

\draw[->, red, dashed] (0.5,0.2) -- (4.4,2.1) 
node [pos=0,below,black] {$\bup[(i-1)]$}
node [below,black] {$\hatn[\ii]$};
\draw[->,blue] (4.4,2.1) -- (5.6,2.3) 
node [pos=0.5,above] {$\tiln[\ii]$}
node [below,black] {$\bup[\ii]$};
\end{tikzpicture}

\caption{Zoom in on a single iteration $i$ including an illustration of the designation of the upper RHS $\bup[(i-1)], \hatn[\ii], \tiln[\ii]$ and $\bup[\ii]$ and the upper bound $D$ of the box sizes and $\|\tiln[\ii]\|_\infty$.}
\label{fig:oneIter}
\end{figure}

\subsection{Structure of Intermediary Solutions}\label{sec:topDown}
In this section, we show that if \eqref{eq:nfold} is feasible, there exists a feasible upper RHS near $\nicefrac{\bup}{2^{\calI-i}}$, for all $i \in [\calI]$, which then can be extended to solve main \nfold. 
We say that a RHS $b$ or upper RHS $\bup$ is feasible if and only if there exists a solution $x$ to $\calA x = (\bup,\bdown)^T$.
We bound the size of the RHS in the small sub-problems needed to be solved in order to extend these intermediary solutions in each step.

Within this section, we view the problem top-down, i.e., in each iteration $i \in [\calI],$ we divide the problem into two sub-problems.
The first sub-problem, which we call \emph{parity constrained}, contains only even entries in its lower RHS $\hatm.$ Because of this, that sub-problem can then be divided in two equal parts which we solve inductively. We call the second sub-problem, with lower RHS $\tilm$, \emph{small}. The \emph{small} sub-problem fulfills $\|\tilm\|_\infty \le K$ and can therefore be solved efficiently using dynamic programming.

For a given iteration $i$, we denote the RHS of the considered problem by $b^\powIteri.$ Note that we have $b^\powIterCalI=b.$ The division into the two sub-problems implies $b^\powIteri =\hat{b}^\powIteri + \tilde{b}^\powIteri.$ As we divide parity constrained problem into two equal problems and then repeat the procedure, we have $b^\powIteriminOne = \nicefrac{\hat{b}^\powIteri}{2}$. We ensure, that the upper RHS $\bup[\ii]$ also contains only even entries.

With this notation in mind, we begin by determining the lower RHS values for each iteration and their sub-problems.
To determine the value of the $k$-th entry in the lower RHS of the small sub-problems $\tilm[\ii]$, we consider the corresponding entry $\bdown[\ii]_k$. If that entry is small, that is if $\bdown[\ii]_k \le K$, we set $\tilm[\ii]_k = \bdown[\ii]_k$. This implies $\hatm[\ii] = \bdown[\ii]_k - \tilm[\ii]_k = 0$ and $\bdown[(i-1)]_k = \nicefrac{\hatm[\ii]_k}{2} = 0$. Therefore, for all following iterations this entry remains 0, i.e,\ $\bdown[(i')]_k = \tilm[(i')]_k = 0$ and $\hatm[(i')]_k = 0$ for all $i' < i$.
However, if $\bdown[\ii]_k$ is large, that is if $\bdown[\ii]_k > K$, we set $\tilm[\ii]_k$ to $K$ or $K-1$ such that $\hatm[\ii]_k = \bdown[\ii]_k-\tilm[\ii]_k$ is even. Thus, the parity constrained entry is larger than 0 and can be divided by 2. Now, set $\bdown[(i-1)]_k = \nicefrac{\hatm[\ii]_k}{2}$ and repeat the procedure for the next smaller iteration. 

More formally, for known value of $\bdown[(i+1)]_k > K$ we can derive $\bdown[\ii]_k$ as follows:
\begin{align}\label{eq:m_k iter}
    \bdown[\ii]_k &= \frac{\bdown[(i+1)]_k - (\maxSupp+z_k^\powIteriplusOne)}{2} \qquad \text{with }\\
    z_k^\powIteriplusOne &= \begin{cases} 
    0, & \text{if } \big(\bdown[(i+1)]_k \mod 2 \big) = (\maxSupp \mod 2) \\
    1, & \text{otherwise}. 
    \end{cases}\nonumber
\end{align}

The following lemmas derive that, despite of the uncertainty due to the parity-constraint, the lower RHS in each iteration can be uniquely determined from the input $(\calA,b)$ to \eqref{eq:nfold}.

\begin{restatable}[\Rightscissors]{lemma}{mkInt}\label{lem:m_k integer}
   Let $\bdown_k \in \ZZgeqzero^n$ and $k\in[n]$. Then $\bdown[\ii]_k$ is integer for all $i \in [\calI],$ where values are determined by the procedure iteratively applying \eqref{eq:m_k iter}.
\end{restatable}

Using this property and \eqref{eq:m_k iter}, we achieve the desired result by elegantly formulating the term as a continued fraction and using a property of the geometric sum. 
\begin{restatable}[\Rightscissors]{lemma}{mkFormula}\label{lem:m_k formula}
    Let $\bdown_k \in \ZZgeqzero^n$ and $k\in[n]$.
    Then $\bdown[\ii]_k, i \in [\calI]$ is determinable as
    \begin{align} \label{eq:m_k^(i)}
        \bdown[\ii]_k = \begin{cases}
            0, & \text{if } i \neq \calI \text{ and } \bdown[(i+1)]_k \leq \maxSupp\\
            \Big\lfloor\nicefrac{\bdown_k}{2^{\calI-i}} - \maxSupp \cdot \sum_{\ell=i}^{\calI-1} \big(\nicefrac{1}{2^{\calI-\ell}}\big)\Big\rfloor, & \text{otherwise.}
        \end{cases}
    \end{align}
\end{restatable}
\begin{proof}[Proof Idea.]
We determine the interval in which $\bdown[\ii]_k$ lies. 
The interval size is less than one. With \cref{lem:m_k integer}, this implies the lemma.
\end{proof}

We now focus on the case, when $\bdown[(i+1)]_k$ is large.
Since we reduce the lower RHS in each entry by at most $\maxSupp$ such that the remaining entries are even, the remaining lower RHS can be divided by two.
However, this does not directly imply that we can split the complete ILP into two equal sub-ILPs, as the upper RHS might have odd entries.
To provide that property, we give a procedure on how to reduce the solution of the \nfold within the iterations.
We show that after the reduction, there exists a feasible solution vector of the remaining problem, which contains only even entries, which then implies that the complete RHS is even.
Again, we make use of the support theorem (\cref{lem:support}) of Eisenbrand and Shmonin~\cite[Theorem 1 (ii)]{ES06}.
Define 
\begin{align} \label{eq:define tilde xi}
    \tilx_k^\powIteri := 
    \begin{cases}
    	x_k^\powIteri, &\text{if } \big\|x_k^\powIteri \big\|_1 \leq \maxSupp\\
    	\Big(\tilx_{k,1}^\powIteri, \dots,\tilx_{k,t_k}^\powIteri\Big)^T, &\text{otherwise}.
    \end{cases}
\end{align}
In the latter case, determine $\Big(\tilx_{k,1}^\powIteri, \dots,\tilx_{k,t_k}^\powIteri\Big)^T$ as follows (for details we refer to \cref{alg:get tilde x}).
\begin{enumerate}
    \item Set $\tilx_k^\powIteri = \nullvec_{t_i}$.
    \item For all odd entries $\ell \in [t_k]$ in $x_{k,\ell}^\powIteri$, add 1 to $\tilx_{k,\ell}^\powIteri$. This ensures the correct parity.
    \item Until $\|\tilx_k^\powIteri\|_1$ equals the desired value $K$ or $K-1$, add 2 to an arbitrary component $\tilx_{k,\ell}^\powIteri$, with $x_{k,\ell}^\powIteri - \tilx_{k,\ell}^\powIteri \ge 0.$ Note that there always exists such entry as otherwise, we would have $\bdown[\ii]_k \le K$ and the RHS of the remaining problem would be $\nullvec$.
\end{enumerate}
This procedure ensures that $x^\powIteri - \tilx^\powIteri =: \hatx^\powIteri$ only consists of even components. Therefore, the RHS $\hat{b}^\powIteri$ of such ILP $\calA \hatx^\powIteri = \hat{b}^\powIteri$ does not have any odd entries. Therefore, the complete problem can be split into two equal sub-problems.

Using the previously described reduction strategies, we can now deduce how the upper RHS $\bup[\ii]$ changes during an iteration. 
Since we do not know the solution $x$ of the \nfold \eqref{eq:nfold}, the intermediate solutions $x^\powIteri$ are also unknown. 
Note that \cref{alg:get tilde x} only gives a procedure to derive the current intermediate solution, assuming the prior one is known.
Therefore, we cannot uniquely determine the upper RHS of the sub-problems.
In the following, we often refer to a feasible RHS $b$ as a feasible upper RHS or point $\bup$ when $\bdown$ is apparent from the context.
Recall that it is enough to determine the feasible points near $\nicefrac{\bup}{2^{\calI-i}}$ for all $i \in [\calI].$ In the following, we determine the size of the boxes in which a feasible point must lie if \eqref{eq:nfold} is feasible.
For this purpose set $D := n \cdot \maxSupp \cdot \Delta.$

The next lemma states that for given small lower RHS $\tilm[\ii]$ the size of the upper RHS $\tiln[\ii]$ is bounded by $D$ in each dimension (\cref{fig:oneIter}).

\begin{restatable}[\Rightscissors]{lemma}{lengthJobs}\label{lem:jobbound}
    Let $\tilde{b}^\powIteri = \big(\tiln[\ii],\tilm[\ii]\big)^T$ with $i \in [\calI], \|\tilm[\ii]\|_\infty \leq \maxSupp$
    and assume $\calA \tilx^\powIteri = \tilde{b}^\powIteri$ is feasible.
    Then it holds that $\big\|\tiln[\ii]\big\|_\infty \leq D.$
\end{restatable}
With this property, we can derive the boxes with all relevant feasible upper RHS of the sub-problems.
The center of the box in iteration $i$ is $\nicefrac{\bup}{2^{\calI-i}}$ and the side length in each dimension is $2D$ (\cref{fig:oneIter}).

\begin{restatable}[\Rightscissors]{lemma}{properties}\label{lem:properties norm n leq D}
    Let $b = (\bup,\bdown)^T$ and assume $\calA x=b$ is feasible. Then the following two properties hold:
    
        \leftskip1em
        (i) \; For all $i \in \{2,\dots,\calI\},$ there exists $\bup[\ii] \in \ZZgeqzero^r$ with  $\big\| \nicefrac{\bup}{2^{\calI-i}} - \bup[\ii] \big\|_\infty \leq D$ and $\calA x^\powIteri = (\bup[\ii], \bdown[\ii])^T$ is feasible.
            
        (ii) \; There exists $\bup[(1)] \in \ZZgeqzero^r$ with $\|\bup[(1)]\|_\infty \leq D$ and $\calA x^\powIterOne=(\bup[(1)],\bdown[(1)])^T.$
\end{restatable}

This means that if there exists a feasible point $\bup[\ii]$ near $\nicefrac{\bup}{2^{\calI-i}},$ there also exists another feasible point $\bup[(i-1)]$ near $\nicefrac{\bup}{2^{\calI-(i-1)}}$ for all iterations $i \in \{3,\dots,\calI\}.$
We also find a feasible point $\bup[(1)]$ near $\nullvec$ if there is a feasible point $\bup[(2)]$ near $\nicefrac{\bup}{2^{\calI-2}}.$
Altogether, this gives us a path from $\nullvec$ to $n$, with at least one feasible point near each $\nicefrac{\bup}{2^{\calI-i}}$ for all $i \in \{2,\dots,\calI\}.$

In the following, we determine the total number of needed iterations the algorithm needs by showing how we can determine the number of iterations until the lower RHS is equal to 0 in each entry.
We define $\calI_k$ as the number of \textit{non-zero} iterations of machine-type $k.$ 
We say that an iteration $i$ of machine-type $k$ is non-zero when $\bdown[\ii]_k \neq 0.$

\begin{restatable}[\Rightscissors]{lemma}{IkFormula}\label{lem:calI_k formula}
    Let $k\in[n]$, $\bdown_k \in \ZZgeqzero.$ 
    Then the number of non-zero iterations $\calI_k$ can be determined by
    \begin{align*}
        \calI_k = \calIkFormula
    \end{align*}
\end{restatable}
\begin{proof}[Proof Idea.]
We use \eqref{eq:m_k^(i)} in \cref{lem:m_k formula} to determine the iteration $i$ with $\bdown[\ii]_k > 0$ and $\bdown[(i-1)]_k = 0.$ As determining $\bdown[\ii]_k$ (\cref{alg:m_k^(i)}) is deterministic, this can be calculated uniquely.
\end{proof}

Having determined the number of non-zero iterations of each machine-type, we can directly derive the total number of iterations $\calI,$ as $\calI = \max_{k \in [n]}\calI_k.$ 
\begin{restatable}[\Rightscissors]{lemma}{IFormula} \label{lem:calI}
    For the total number of iterations $\calI$ it holds that
    \begin{align*}
        \calI = \calIFormula
    \end{align*}
\end{restatable}

With this in mind, we are now able to derive an algorithm where we start with a small \nfold ILP in the first iteration and build up the overall solution of \eqref{eq:nfold} step by step.

\subsection[Solving the n-fold ILP]{Solving the \nfold ILP}
In this section, we introduce an algorithm that determines the feasibility of the \nfold ILP \eqref{eq:nfold} with only non-negative entries in $\calA$. Combined with the reduction in \cref{sec:reduction}, this results in the following:
\begin{restatable}{lemma}{mainthm}\label{lem:main}
    The feasibility of combinatorial \nfold ILPs \eqref{eq:nfold} where $\calA$ contains only non-negative coefficients and $\Delta$ is the largest coefficient in $\calA$ can be determined in time $\RTfeasibilityO$.
\end{restatable}

\begin{proof}[Proof Idea.]
In the preceding section we derived the property that we only have to solve small sub-problems (\cref{lem:jobbound}) and that we can combine their solutions in an efficient way to solve \eqref{eq:nfold}.
We split the algorithm into two parts. In the preprocessing we generate solutions for all required small sub-problems. Then, we generate the final solution by combining the intermediate solutions.
A detailed proof of both, the dynamic program and the combining step can be found in the appendix.

\paragraph*{Preprocessing}
First, determine the number of non-zero iterations $\calI_k$ for all $k \in [n]$ and the total number of iterations $\calI$ with \cref{lem:calI_k formula} and \cref{lem:calI}.
Then determine all required lower RHS $\bdown[\ii], \tilm[\ii]$ and $\hatm[\ii]$ for all $i \in [\calI]$ as described in \cref{sec:topDown} (and \cref{alg:m_k^(i)}).
For each vector $\tilm[\ii]$ define the set of required upper RHS in iteration $i:$
\begin{align*}
    \tilde{N}^\powIteri = \Big\{\tiln[\ii] \in \ZZgeqzero^r \mid \calA \tilx^\powIteri = \big(\tiln[\ii],\tilm[\ii]\big)^T, \tilx^\powIteri \in \ZZgeqzero^h, \tiln[\ii] \leq D\Big\}.
\end{align*}

The elements in $\tilde{N}^\powIteri$ can be calculated via dynamic programming (\cref{alg:dyn prog}). Instead of determining the feasibility of the complete \nfold $\calA \tilx^\powIteri = \big(\tiln[\ii],\tilm[\ii]\big)^T$, we split it into its bricks. For each brick $k \in [n]$ and upper RHS $\nu \in \{0,\dots,K\Delta\}^r$, we determine the feasibility of $A_k \tilx_k = \nu$ with $\|\tilx_k\|_1 = \tilm[\ii]_k$. These feasible points are then combined into the set $\tilde{N}^\powIteri$.

\begin{algorithm}
\caption{dynamicProgram}
\algcomment{\textbf{\cref{alg:dyn prog}:} The base table (BT) contains the feasibility of all possible points for each individual block. The dynamic table $DT$ combines the feasibility of the sub-problems. The set of required points is returned.}

\textbf{Input:} Lower RHS $\tilm[\ii]$ \\
\textbf{Output:} Set $\tilde{N}^\powIteri$ of required points (upper RHS)

\begin{algorithmic}[1]
\For{all $k \in [n]$ and $\nu \in \{0, \dots,K\Delta\}^r$} 
\Comment{Base Case}
    \State $\vcenter{
    \begin{flalign*}
        &BT(\nu,k ) \gets \begin{cases} 
  \trueCapital, & \text{if } A_k \tilx_k = \nu \text{ with } \|\tilx_k\|_1 = \tilm[\ii]_k, \tilx_k \in \ZZgeqzero^{t_i} \text{ is feasible} \\ 
  \falseCapital, & \text{otherwise } 
  \end{cases}&
  \end{flalign*}}$
\EndFor
\State $\vcenter{
    \begin{flalign*}
        & DT(\nu, 1 ) \gets BT(\nu,1)&
    \end{flalign*}}$
    
\For{all $k = 2,\dots,n$ and $\nu \in \{0, \dots,D\}^r$} \Comment{Inductive Step}
    \State $\vcenter{
    \begin{flalign*}
        & DT(\nu, k ) \gets \bigvee_{\substack{\nu'+\nu'' = \nu,\\
    \nu',\nu'' \text{ integral vectors}}} \big(DT(\nu',k-1) \land  BT(\nu'',k)\big)&
    \end{flalign*}}$
\EndFor
\Return $\big\{\nu \in \{0, \dots,D\}^r \mid DT(\nu,n) = \trueCapital \big\}$
\end{algorithmic}

\label{alg:dyn prog}
\end{algorithm}
The feasibility of the small sub-problems (base case) can be determined using a standard ILP algorithm. Applying the currently Jansen-Rohwedder algorithm~\cite{JR23} yields a running time of 
\begin{align}\label{eq:pcmax rt}
    O(\sqrt{r+1} \Delta)^{(1+o(1))(r+1)} + O(t_k(r+1)).
\end{align}

\paragraph*{Combining the Solutions}
Having solved all small sub-problems, we are now able to iteratively derive the feasibility of \eqref{eq:nfold}.
In the first iteration, set $N^\powIterOne := \tilde{N}^\powIterOne$ as the set of feasible points after this iteration. Note that for all $\bup[(1)] \in N^\powIterOne$ it holds that $\calA x^\powIterOne=(\bup[(1)],\bdown[(1)])^T$ is feasible.

In every other iteration $i = 2,\dots,\calI,$ we first double the solutions of the previous iteration and obtain the feasible ILP $\calA \hatx^\powIteri=(\hatn[\ii],\hatm[\ii])^T,$ with $\hatx^\powIteri = 2 x^\powIteriminOne, \hatm[\ii] = 2 \bdown[(i-1)]$ and $\hatn[\ii] = 2 \bup[(i-1)]$.
Let $\hat{N}^\powIteri := \big\{\hatn[\ii] \in \ZZgeqzero^r \mid \hatn[\ii] = 2 \bup[(i-1)], \bup[(i-1)] \in N^\powIteriminOne\big\}$ be the set of feasible points after that step. 
We now construct the set $N^\powIteri$ of the feasible points after this iteration by combining all vectors in $\hat{N}^\powIteri$ with all vectors in $\tilde{N}^\powIteri$ and only keep the ones close enough to $\nicefrac{\bup}{2^{\calI-i}}$ (\cref{lem:properties norm n leq D}), i.e.,
\begin{align*}
    N^\powIteri := \big\{\bup[\ii] \in \ZZgeqzero^r \mid \bup[\ii] = \hatn[\ii] + \tiln[\ii], \hatn[\ii] \in \hat{N}^\powIteri, \tiln[\ii] \in \tilde{N}^\powIteri, \big\|\nicefrac{\bup}{2^{\calI-i}} - \bup[\ii]\big\|_\infty \leq D\big\}.
\end{align*}
Note again that for all $\bup[\ii] \in N^\powIteri$ it holds that $\calA x^\powIteri=(\bup[\ii],\bdown[\ii])^T$ is feasible.

Repeating these steps results in a set $N^\powIterCalI$ of feasible points with 
$\calA x^\powIterCalI = \big(\bup[(\calI)],\bdown[(\calI)] \big)^T$ and $\bdown[(\calI)] = \bdown, \bup[(\calI)] \in N^\powIterCalI.$ 
Therefore, if $\bup \in N^\powIterCalI$ the \nfold ILP \eqref{eq:nfold} is feasible.

\paragraph*{Running Time}
The running time of the preprocessing is dominated by the dynamic program. 
In the worst case, it is called $\Oh(n)$ times.
Creating the base table (BT) of the DP, we have to solve $n (K\Delta+1)^r$ small sub-problems, each having a running time of \eqref{eq:pcmax rt}.
The dynamic program then constructs a dynamic table (DT) with $n(D+1)^r$ entries. Each can be calculated by doing a boolean convolution using FFT which takes time $\Oh(n(D+1)^r \cdot \log(n(D+1)^r))$ (see e.g.~\cite{BN21}).
With $t_k \le h$ for all $k\in [n]$, the total running time for the preprocessing amounts to
\begin{align*}
    &\underbrace{n (K\Delta+1)^r \cdot O(\sqrt{r+1} \Delta)^{(1+o(1))(r+1)} + O(h(r+1))}_\text{Base Case} + \underbrace{n^2 (D+1)^\Or \cdot \log(n (D+1)^r)}_\text{Inductive Step}\\
    &= (D+1)^\Or
\end{align*}

The number of iterations, we need to combine the solutions is generally bounded by $\log(\bdown_{\max})$. However, if the upper RHS is relatively small in all entries, it becomes $\nullvec$ before $\bdown$ does. In order to be feasible in that case, we only need to check whether a $\nullvec$ in each $A_i$ where the corresponding entry in the lower RHS is not 0, exists. Thus, these sub-problems are trivial to solve and therefore, we may say that we need $\log(b_{\text{def}})$ iterations, where $b_{\text{def}} := \min(\bup_{\max},\bdown_{\max})$. Each set $N^\powIteri, \tilde{N}^\powIteri$ and $\hat{N}^\powIteri$ contains at most $(D+1)^r$ vectors. Therefore, determining each set is possible in time $D^\Or.$
With $h \le (\Delta+1)^r$ this yields a total running time of $\RTfeasibilityO$.
\end{proof}

\subsection{Allowing Negative Coefficients}\label{sec:reduction}
In this section, we provide a linear-time reduction from a general combinatorial \nfold ILP \eqref{eq:nfold} to an ILP of the same form which only allows non-negative entries in $\calA.$ More concretely, we prove:
\begin{lemma}
    Let $X \subseteq \ZZgeqzero^h$ be the set of feasible solutions of the \nfold ILP \eqref{eq:nfold} with sub-matrices $A_i \in \ZZ^{r \times t_i} \forall i \in [n]$ and RHS vector $b \in \ZZ^{r+n}$ and let $\Delta$ be the largest absolute value in $A_i$, i.e., $\Delta = \max_{i \in [n]} \|A_i\|_\infty$. Then there are matrices  $A_i' \in \ZZgeqzero^{r\times t_i}  \forall i \in [n]$ and a vector $b' \in \ZZgeqzero^{r+n}$ such that $X$ is the set of feasible solutions to  $\calA' x = b'$ of the form \eqref{eq:nfold} with sub-matrices $A_i$ and the entries in $A_i'$ are bounded by $\Delta' = 2\Delta$. 
\end{lemma}
\begin{proof}
    Let $i\in[n]$. We construct $A_i'$ as by adding $\Delta$ to each entry in $A_i$, i.e., $A_i'[k,\ell]:=A_i[k,\ell]+\Delta$, where $A_i[k,\ell]$ is the entry in $A_i$ that is in the $k$-th row and $\ell$-th column (analogous notation for $A_i'[k,\ell]$). 
    Since the absolute entries in $A_i$ are smaller than or equal to $\Delta,$ all entries in $A_i'$ are non-negative and upper bounded by $2\Delta =: \Delta'$.

    Now we show how the RHS is adapted. Note that the local constraints of \eqref{eq:nfold} define the 1-norm of the solution vectors, i.e., for all $x \in X$ it holds that $\|x\|_1 = \|\bdown\|_1$. Viewing the matrix-multiplication as summing up vectors (columns of the matrix), we are summing up exactly $\|x\|_1$ vectors.
    Therefore, by changing the constraint matrix as described above, we know that exactly $\|\bdown\|_1 \cdot \Delta$ are added to each entry of the upper RHS. Thus, we define $\bup[']$ as follows: For each $j \in[r]$ set $\bup[']_j := \bup + \|\bdown\|_1 \cdot \Delta$. 
    Since the local constraints do not change, the lower RHS stays the same that is $\bdown['] := \bdown$.

\end{proof}
\subsection{Solving the Optimization Problem}
\label{subsec:Optimization}
While we focused our discussions on the problem of deciding whether a feasible solution to a given combinatorial \nfold exists, we note that some modifications further extend our results to finding an optimal solution according to some objective function, for instance $\max c^Tx, c \in \ZZ^h$. We need to adapt the upper bound on the support (\cref{lem:support}) and the dynamic program (\cref{alg:dyn prog}) according to this new objective function. For $\max c^Tx, c \in \ZZ^h$, the running time of the dynamic program changes as follows: Instead of determining the feasibility of the small sub-problems in the base case, we need the optimal solution of those. The currently best algorithm by Jansen and Rohwedder~\cite{JR23} achieves a running time of $O(\sqrt{r+1} \Delta)^{2(r+1)} + O(h(r+1))$. In the inductive step, we replace the boolean convolution with $(\max,+)$-convolution, which has generally has quadratic running time. 
This and the change of the support bound to $2(r+2)(\log(r+2)+\Delta+2)$ does not affect the asymptotic running time which yields:
\thmopt*

\section{Applying the Algorithm}
\label{sec:Applications}
To complete the discussion of our algorithm, we present exemplary applications of the algorithm to different contexts. Here, we introduce the considered problems, while we refer to \cref{sec:appendix-sched} for the detailed applications and resulting running times.

\paragraph*{Scheduling on Uniform Machines}
A scheduling instance is defined by the number of jobs $n \in \ZZgeqzero^d$ with processing times $p \in \ZZgeqzero^d$ and the number of machines $m \in \ZZgeqzero^\tau$ with speed values $s \in \ZZ_{\ge 1}^\tau.$
This means that $N:=\|n\|_1$ jobs of $d$ different sizes have to be scheduled on $M:=\norm{m}_1$ machines with $\tau_I$ different speeds. Since the processing times are integer, the total number of different job-types $d$ is bounded by $d \leq \pmax.$
Each job has to be executed on a single machine.

Let $\mathcal{M}$ be the set of all machines.
A \textit{schedule} $\sigma \colon \ZZgeqzero^d \to \mathcal{M}$ is a mapping that assigns all jobs to the available machines. 
The \textit{load} $L(m^{(k)})^{(\sigma)}$ of machine $m^{(k)} \in \mathcal{M}$ is the sum of the sizes of all jobs assigned to that machine by the schedule $\sigma.$
Assume $m^{(k)}$ has speed $s_k.$ 
Then we define the \textit{completion time} of machine $m^{(i)}$ by $C(m^{(i)})^{(\sigma)} := \nicefrac{L(m^{(i)})^{(\sigma)}}{s_k},$ i.e., the time when machine $m^{(i)}$ has finished processing all jobs that were assigned to it by $\sigma.$

The objective functions that we consider in this work are:
\begin{itemize}
    \item \textbf{Makespan Minimization ($\Cmax$):} $
            \min_\sigma \Big(\max_{m^{(i)} \in \mathcal{M}} C(m^{(i)})^{(\sigma)} \Big)$
    \item \textbf{Santa Claus ($\Cmin$):} $
            \max_\sigma \Big(\min_{m^{(i)} \in \mathcal{M}} C(m^{(i)})^{(\sigma)}\Big).$
\end{itemize}

\paragraph*{Closest String Problem}
When considering the \closeststring problem, we are given a set of $k$ strings $S = \{s_1,\dots,s_k\}$ of length $L \in \ZZ_{\ge1}$ from alphabet $[k]$ and an integer $d \in \ZZgeqzero$. The task is to find a string $c \in [k]^L$ of length $L$ (called the "closest string") such that the Hamming distance $d_H(c,s_i)$ between $c$ and any string $s_i \in S$ is at most $d$. The Hamming distance for two equal-length strings is the number of positions at which they differ.

\paragraph*{Imbalance Problem}
The input to the \imbalance problem is an undirected graph $G=(V,E)$ with $n$ vertices. An ordering of the vertices in $G$ is a bijective function $\pi=V\rightarrow[n]$. For $v\in V$, we define neighborhood of $v$ by $N(v):=\{u \in V \mid (u,v) \in E\}$, the set of left neighbors $L(v):=\{u\in N(v) \mid \pi(u)<\pi(v)\}$ and the set of right neighbors $R(v):=\{u\in N(v) \mid \pi(v)<\pi(u)\}$. Note that $R(v) = N(v)\setminus L(v)$. The imbalance of at $v\in V$ is defined as $\iota_\pi(v) = ||R(v)|-|L(v)||$ and the total imbalance of the ordering $\pi$ is $\iota(\pi) = \sum_{v\in V} \iota_\pi(v)$.
The aim is now to find an ordering $\pi$ that minimizes the total imbalance.

\section{Conclusion}
We have shown that we can solve combinatorial \nfold ILPs in time $\RTfeasibilityO.$ Our algorithm builds on insights regarding the partition of the entire problem into smaller and smaller sub-problems, which we solve and use to reconstruct the entire problem in a novel manner. Complementing our algorithmic results, we present applications of our algorithm to the world of scheduling on uniform machines, the closest string problem and imbalance. In the case of scheduling on uniform machines, we can apply techniques to bound $\Delta$ and achieve an algorithm with a running time matching a lower bound provided by the ETH for both makespan minimization and Santa Claus. For closest string our algorithm matches the current state of the art, but can, in contrast to existing algorithms, leverage bounded number of column-types to achieve improved running times. Finally, when applied to imbalance, our algorithm results in an improved running time over the state-of-the-art.
For both closest string and imbalance our algorithm matches the current state of the art, but can, in contrast to existing algorithms, leverage bounded number of column/vertex types to achieve improved running times.
\bibliography{bib}

\begin{thebibliography}{10}

\bibitem{AschenbrennerH07}
Matthias Aschenbrenner and Raymond Hemmecke.
\newblock Finiteness theorems in stochastic integer programming.
\newblock {\em Foundations of Computational Mathematics}, 7(2):183--227, 2007.

\bibitem{Bakken18}
Olav~R{\o}the Bakken.
\newblock Arrangement problems parameterized by neighbourhood diversity.
\newblock {\em Master’s thesis, University of Bergen}, 2018.

\bibitem{BJMS23}
Sebastian Berndt, Hauke Brinkop, Klaus Jansen, Matthias Mnich, and Tobias
  Stamm.
\newblock New support size bounds for integer programming, applied to makespan
  minimization on uniformly related machines.
\newblock In {\em International Symposium on Algorithms and Computation
  ({ISAAC})}, volume 283 of {\em LIPIcs}, pages 13:1--13:18, 2023.

\bibitem{BiedlCGHW05}
Therese Biedl, Timothy~M. Chan, Yashar Ganjali, Mohammad~Taghi Hajiaghayi, and
  David~R. Wood.
\newblock Balanced vertex-orderings of graphs.
\newblock {\em Discrete Applied Mathematics}, 148(1):27--48, 2005.

\bibitem{BN21}
Karl Bringmann and Vasileios Nakos.
\newblock Fast n-fold boolean convolution via additive combinatorics.
\newblock In {\em International Colloquium on Automata, Languages, and
  Programming, (ICALP)}, volume 198 of {\em LIPIcs}, pages 41:1--41:17, 2021.

\bibitem{ChenJZ18}
Lin Chen, Klaus Jansen, and Guochuan Zhang.
\newblock On the optimality of exact and approximation algorithms for
  scheduling problems.
\newblock {\em Journal of Computer and System Sciences}, 96:1--32, 2018.

\bibitem{ChenMYZ17}
Lin Chen, D{\'{a}}niel Marx, Deshi Ye, and Guochuan Zhang.
\newblock Parameterized and approximation results for scheduling with a low
  rank processing time matrix.
\newblock In {\em Symposium on Theoretical Aspects of Computer Science,
  {(STACS)}}, volume~66 of {\em LIPIcs}, pages 22:1--22:14, 2017.

\bibitem{CEHRW21}
Jana Cslovjecsek, Friedrich Eisenbrand, Christoph Hunkenschr{\"{o}}der, Lars
  Rohwedder, and Robert Weismantel.
\newblock Block-structured integer and linear programming in strongly
  polynomial and near linear time.
\newblock In {\em ACM-SIAM Symposium on Discrete Algorithms (SODA)}, pages
  1666--1681, 2021.

\bibitem{CslovjecsekEPVW21}
Jana Cslovjecsek, Friedrich Eisenbrand, Michal Pilipczuk, Moritz Venzin, and
  Robert Weismantel.
\newblock Efficient sequential and parallel algorithms for multistage
  stochastic integer programming using proximity.
\newblock In {\em European Symposium on Algorithms, {ESA}}, volume 204 of {\em
  LIPIcs}, pages 33:1--33:14, 2021.

\bibitem{CyganFKLMPPS15}
Marek Cygan, Fedor~V. Fomin, Lukasz Kowalik, Daniel Lokshtanov, D{\'{a}}niel
  Marx, Marcin Pilipczuk, Michal Pilipczuk, and Saket Saurabh.
\newblock {\em Parameterized Algorithms}.
\newblock Springer, 2015.

\bibitem{EHK18}
Friedrich Eisenbrand, Christoph Hunkenschr{\"{o}}der, and Kim{-}Manuel Klein.
\newblock Faster algorithms for integer programs with block structure.
\newblock In {\em International Colloquium on Automata, Languages, and
  Programming, (ICALP)}, volume 107 of {\em LIPIcs}, pages 49:1--49:13, 2018.

\bibitem{EHKKLO22}
Friedrich Eisenbrand, Christoph Hunkenschröder, Kim-Manuel Klein, Martin
  Koutecký, Asaf Levin, and Shmuel Onn.
\newblock An algorithmic theory of integer programming, 2022.
\newblock \href {https://arxiv.org/abs/1904.01361} {\path{arXiv:1904.01361}}.

\bibitem{ES06}
Friedrich Eisenbrand and Gennady Shmonin.
\newblock Carath{\'{e}}odory bounds for integer cones.
\newblock {\em Operations Research Letters}, 34(5):564--568, 2006.

\bibitem{FellowsLMRS08}
Michael~R. Fellows, Daniel Lokshtanov, Neeldhara Misra, Frances~A. Rosamond,
  and Saket Saurabh.
\newblock Graph layout problems parameterized by vertex cover.
\newblock In {\em International Symposium on Algorithms and Computation
  ({ISAAC})}, volume 5369 of {\em LNCS}, pages 294--305, 2008.

\bibitem{GavenciakKK22}
Tomas Gavenciak, Martin Kouteck{\'{y}}, and Dusan Knop.
\newblock Integer programming in parameterized complexity: Five miniatures.
\newblock {\em Discrete Optimization}, 44(Part):100596, 2022.

\bibitem{GNR03}
Jens Gramm, Rolf Niedermeier, and Peter Rossmanith.
\newblock Fixed-parameter algorithms for {CLOSEST} {STRING} and related
  problems.
\newblock {\em Algorithmica}, 37(1):25--42, 2003.

\bibitem{HemmeckeOR13}
Raymond Hemmecke, Shmuel Onn, and Lyubov Romanchuk.
\newblock $n$-fold integer programming in cubic time.
\newblock {\em Mathematical Programming}, 137(1-2):325--341, 2013.

\bibitem{ImpagliazzoP01}
Russell Impagliazzo and Ramamohan Paturi.
\newblock On the complexity of k-sat.
\newblock {\em Journal of Computer and System Sciences}, 62(2):367--375, 2001.

\bibitem{JansenKZ25}
Klaus Jansen, Kai Kahler, and Esther Zwanger.
\newblock Exact and approximate high-multiplicity scheduling on identical
  machines.
\newblock In {\em Conference on Algorithms and Complexity ({CIAC})}, volume
  15679 of {\em LNCS}, pages 1--17, 2025.

\bibitem{JansenKMR22}
Klaus Jansen, Kim{-}Manuel Klein, Marten Maack, and Malin Rau.
\newblock Empowering the configuration-{IP}: new {PTAS} results for scheduling
  with setup times.
\newblock {\em Mathematical Programming}, 195(1):367--401, 2022.

\bibitem{JansenLR20}
Klaus Jansen, Alexandra Lassota, and Lars Rohwedder.
\newblock Near-linear time algorithm for $n$-fold {ILP}s via color coding.
\newblock {\em {SIAM} Journal on Discrete Mathematics}, 34(4):2282--2299, 2020.

\bibitem{JR23}
Klaus Jansen and Lars Rohwedder.
\newblock On integer programming, discrepancy, and convolution.
\newblock {\em Mathematics of Operations Research}, 48(3):1481--1495, 2023.

\bibitem{Lenstra83}
Hendrik W.~Lenstra Jr.
\newblock Integer programming with a fixed number of variables.
\newblock {\em Mathematics of Operations Research}, 8(4):538--548, 1983.

\bibitem{KaraKW07}
Jan K{\'{a}}ra, Jan Kratochv{\'{\i}}l, and David~R. Wood.
\newblock On the complexity of the balanced vertex ordering problem.
\newblock {\em Discrete Mathematics \& Theoretical Computer Science}, 9(1),
  2007.

\bibitem{Karp72}
Richard~M. Karp.
\newblock Reducibility among combinatorial problems.
\newblock In Raymond~E. Miller, James~W. Thatcher, and Jean~D. Bohlinger,
  editors, {\em Complexity of Computer Computations}, pages 85--103. Springer
  US, Boston, MA, 1972.

\bibitem{KnopK18}
Dusan Knop and Martin Kouteck{\'{y}}.
\newblock Scheduling meets $n$-fold integer programming.
\newblock {\em Journal of Scheduling}, 21(5):493--503, 2018.

\bibitem{KnopKLMO21}
Dusan Knop, Martin Kouteck{\'{y}}, Asaf Levin, Matthias Mnich, and Shmuel Onn.
\newblock Parameterized complexity of configuration integer programs.
\newblock {\em Operations Research Letters}, 49(6):908--913, 2021.

\bibitem{KKLMO23}
Dusan Knop, Martin Kouteck{\'{y}}, Asaf Levin, Matthias Mnich, and Shmuel Onn.
\newblock High-multiplicity n-fold {IP} via configuration {LP}.
\newblock {\em Mathematical Programming}, 200(1):199--227, 2023.

\bibitem{KnopKM20}
Dusan Knop, Martin Kouteck{\'{y}}, and Matthias Mnich.
\newblock Combinatorial $n$-fold integer programming and applications.
\newblock {\em Mathematical Programming}, 184(1):1--34, 2020.

\bibitem{KK18}
Dušan Knop and Martin Koutecký.
\newblock Scheduling {Meets} {N}-{Fold} {Integer} {Programming}.
\newblock {\em Journal of Scheduling}, 21(5):493--503, October 2018.

\bibitem{KouteckyLO18}
Martin Kouteck{\'{y}}, Asaf Levin, and Shmuel Onn.
\newblock A parameterized strongly polynomial algorithm for block structured
  integer programs.
\newblock In {\em International Colloquium on Automata, Languages, and
  Programming, (ICALP)}, volume 107 of {\em LIPIcs}, pages 85:1--85:14, 2018.

\bibitem{KZ20}
Martin Kouteck{\'{y}} and Johannes Zink.
\newblock Complexity of scheduling few types of jobs on related and unrelated
  machines.
\newblock {\em Journal of Scheduling}, 28(1):139--156, 2025.

\bibitem{DeLoera2008}
Jes{\'{u}}s A.~De Loera, Raymond Hemmecke, Shmuel Onn, and Robert Weismantel.
\newblock ${N}$-fold integer programming.
\newblock {\em Discrete Optimization}, 5(2):231--241, 2008.

\bibitem{LokshtanovMS13}
Daniel Lokshtanov, Neeldhara Misra, and Saket Saurabh.
\newblock Imbalance is fixed parameter tractable.
\newblock {\em Information Processing Letters}, 113(19-21):714--718, 2013.

\bibitem{MisraM21}
Neeldhara Misra and Harshil Mittal.
\newblock Imbalance parameterized by twin cover revisited.
\newblock {\em Theoretical Computer Science}, 895:1--15, 2021.

\bibitem{Papadimitriou81}
Christos~H. Papadimitriou.
\newblock On the complexity of integer programming.
\newblock {\em Journal of the {ACM}}, 28(4):765--768, 1981.

\bibitem{ReisR23}
Victor Reis and Thomas Rothvoss.
\newblock The subspace flatness conjecture and faster integer programming.
\newblock In {\em IEEE Symposium on Foundations of Computer Science (FOCS)},
  pages 974--988, 2023.

\bibitem{rohwedder2025}
Lars Rohwedder.
\newblock Eth-tight {FPT} algorithm for makespan minimization on uniform
  machines.
\newblock In {\em International Colloquium on Automata, Languages, and
  Programming, (ICALP)}, volume 334 of {\em LIPIcs}, pages 126:1--126:13.
  Schloss Dagstuhl - Leibniz-Zentrum f{\"{u}}r Informatik, 2025.

\bibitem{RohwedderW25}
Lars Rohwedder and Karol Wegrzycki.
\newblock Fine-grained equivalence for problems related to integer linear
  programming.
\newblock In {\em Innovations in Theoretical Computer Science ({ITCS})}, volume
  325 of {\em LIPIcs}, pages 83:1--83:18, 2025.

\end{thebibliography}

\appendix

\section{Algorithm \-- Omitted details}\label{sec:appendix-algo}
In this appendix, we present omitted/shortened proofs and algorithmic subroutines. 
First, we present the full proofs of the statements that lead to our combinatorial \nfold algorithm. 
We give an algorithm (\cref{alg:m_k^(i)}) that shows the behavior of the lower RHS throughout the iterations of our procedure.

\begin{algorithm} 
\caption{getNumMachines}
\algcomment{\textbf{\cref{alg:m_k^(i)}:} Determines the entries in the lower RHS of each sub-problem $\bdown[\ii]_k, \tilm[\ii]_k, \hatm[\ii]_k$ in each iteration $i \in [\calI],$ given the entry $\bdown_k$ of the initial lower RHS.}

\textbf{Input:} Number of rows $r \in \ZZgeqzero$ in the global constraints, $k$-th entry $\bdown_k \in \ZZgeqzero$ of the RHS

\begin{algorithmic}[1]
\State $\Delta \gets \|\calA\|_\infty$
\State $\maxSupp \gets \maxSuppFormula$ 
\State $\bdown[(\calI)]_k \gets \bdown_k$ 
    \Comment{Set $\bdown[(\calI)]_k$}
\For{$i = \calI-1, \dots, 1$} \Comment{Iteratively determine other $\bdown[\ii]_k,\tilm[\ii]_k$ and $\hatm[\ii]_k$}
    \If{$\bdown[(i+1)]_k \leq \maxSupp$}
    \Comment{Zero iteration}
        \State $\tilm[(i+1)]_k \gets \bdown[(i+1)]_k$
        \State $\hatm[(i+1)]_k \gets 0$
        \State $\bdown[\ii]_k \gets 0$
    \Else
    \Comment{Non-zero iteration}
        
        \State  $z_k^\powIteriplusOne \gets \begin{cases} 
                0, & \txtif \big(\bdown[(i+1)]_k \mod 2\big) = (\maxSupp \mod 2) \\
                1, & \txtother
                \end{cases}$
                \Comment{Ensures integrity of $\bdown[\ii]_k$}
                \State $\tilm[(i+1)]_k \gets \maxSupp - z_k^\powIteriplusOne$
        \State $\hatm[(i+1)]_k \gets \bdown[(i+1)]_k - \tilm[(i+1)]_k$,
        \State $\bdown[\ii]_k \gets \nicefrac{\hatm[(i+1)]_k}{2}$
    \EndIf
\EndFor
\If{$\calI \ge 2$}
    \State $\tilm[(1)]_k \gets \bdown[(2)]_k$
\EndIf
\State $\hatm[(1)]_k \gets 0$
\end{algorithmic}
\label{alg:m_k^(i)}
\end{algorithm}
In each step, the entries in $\bdown[\ii]$ are integer and we can determine them only knowing $\bdown$. 

\mkInt*
\begin{proof}
    (By induction over the number of iterations.)

    \textbf{Base Case:} Assume $i = \calI.$ Then $\bdown[(\calI)]_k = \bdown_k.$ Since $\bdown_k$ is integer, $\bdown[(\calI)]_k$ is integer as well.

    \textbf{Induction Step:} Let $i \in [\calI-1]$ and assume $\bdown[(\ell)]_k$ is integer for all $\ell \in \{i+1,\dots,\calI\}.$

\textit{Case 1:} Assume $\bdown[(i+1)]_k \leq \maxSupp.$ Then $\bdown[(i)]_k = 0$ is integer.

 \textit{Case 2:} Assume $\bdown[(i+1)]_k > \maxSupp.$ 
            Since $\bdown[\ii]_k=\nicefrac{\hatm[(i+1)]_k}{2},$ we have to show that 
            \begin{align*}
                \hatm[(i+1)]_k &= \bdown[(i+1)]_k-\tilm[(i+1)]_k\\
                &= \bdown[(i+1)]_k - \maxSupp + z_k^\powIteriplusOne
            \end{align*}
            is even, with $z_k^\powIteriplusOne$ being defined as
            \begin{align*}
                z_k^\powIteriplusOne := \begin{cases} 
                0, & \text{if } \big(\bdown[(i+1)]_k \mod 2 \big) = (\maxSupp \mod 2) \\
                1, & \text{otherwise}. 
                \end{cases}
            \end{align*}
            On the one hand, if $\bdown[(i+1)]_k$ and $\maxSupp$ are both even or both odd, then $z_k^\powIteriplusOne = 0$ and $\hatm[(i+1)]_k = \bdown[(i+1)]_k-\maxSupp$ is even. 
            On the other hand, if either $\bdown[(i+1)]_k$ or $\maxSupp$ is even and the other value is odd, then $\bdown[(i+1)]_k-\maxSupp$ is odd and $\hatm[(i+1)]_k = \bdown[(i+1)]_k-\maxSupp+z_k^\powIteriplusOne  = \bdown[(i+1)]_k-\maxSupp+1$ is even.
            
            This implies that $\bdown[\ii]_k = \nicefrac{\hatm[(i+1)]_k}{2}$ is integer for all $i \in [\calI]$.
\end{proof}

\mkFormula*
\begin{proof}

We split the proof into two parts. First, we show that \eqref{eq:m_k^(i)} holds for $i = \calI$. Then, we derive the formula for $i < \calI$ using the knowledge of the first part.

\textbf{Part 1:} Assume $i = \calI.$ Then
\begin{align*}
    \bdown[\ii]_k &= \bdown[(\calI)]_k\\
    &= \Bigg\lfloor\frac{\bdown_k}{2^{\calI-\calI}} - \maxSupp \cdot \sum_{\ell=\calI}^{\calI-1} \Big( \frac{1}{2^{\calI-\ell}}\Big)\Bigg\rfloor \\
    &= \lfloor \bdown_k \rfloor.
\end{align*}
Since $\bdown_k$ is integer, we get $\bdown[(\calI)]_k = \bdown_k.$

\textbf{Part 2:} Let $i \in [\calI-1].$

\textit{Case 1:} Assume $\bdown[(i+1)]_k \leq \maxSupp.$ Then, by definition, it holds that $\bdown[\ii]_k = 0.$

\textit{Case 2:} Assume $\bdown[(i+1)]_k > \maxSupp.$
    When $\bdown[(i+1)]_k$ is known and $\bdown[(i+1)]_k > \maxSupp,$ then $\bdown[\ii]_k$ has exactly two possible values, depending on $z_k^\powIteriplusOne,$ i.e., the parity of $\maxSupp$ and $\bdown[(i+1)]_k.$
    Since only $\bdown[(\calI)]_k$ is given, $\bdown[\ii]_k$ can take one of at most $2^{\calI-i}$ values which can be formulated as the following continued fraction:

\begin{align*}
    \bdown[\ii]_k 
    &= \frac{\bdown[(i+1)]_k - \big(\maxSupp+z_k^\powIteriplusOne \big)}{2} &\\
    &= \frac{\cfrac{\bdown[(i+2)]_k - \big(\maxSupp+z_k^{(i+2)} \big)}{2} - \big(\maxSupp+z_k^\powIteriplusOne\big)}{2} &\\
    &\qquad\cfrac{\cfrac{\bdown[(\calI)]_k-\big(\maxSupp+z_k^\powIterCalI \big)}{2} -\big(\maxSupp+z_k^{(\calI-1)}\big)}{2} \\
    &= \qquad\qquad\qquad\qquad\qquad\qquad\qquad\qquad\qquad\ddots\qquad\qquad \\
    &\qquad \frac{\qquad\qquad\qquad\qquad\qquad\qquad\qquad\qquad\qquad\qquad - \big(\maxSupp + z_k^\powIteriplusOne\big)}{2}\\
    &= \frac{\bdown[(\calI)]_k} {2^{\calI-i}} - \frac{\maxSupp+z_k^\powIterCalI} {2^{\calI-i}} - \frac{\maxSupp+z_k^{(\calI-1)}}{2^{\calI-i-1}} - \cdots - \frac{\maxSupp+z_k^\powIteriplusOne}{2^1}\\
    &= \frac{\bdown_k}{2^{\calI-i}} - \sum_{\ell=i}^{\calI-1} \bigg(\frac{\maxSupp+z_k^{(\calI-(\ell-i))}} {2^{\calI-\ell}}\bigg)
\end{align*}

Note that $\bdown[\ii]_k$ is minimal if $z_k^{(\calI - (\ell-i))} = 1$ for all $\ell \in \{i,\dots, \calI-1\}.$ 
Similarly, $\bdown[\ii]_k$ is maximal if $z_k^{(\calI - (\ell-i))} = 0$ for all $\ell \in \{i,\dots, \calI-1\}.$ 
Therefore, 

\begin{align*}
	\bdown[\ii]_k 
    &\in \Bigg[\frac{\bdown_k}{2^{\calI-i}} - \sum_{\ell=i}^{\calI-1} \Big(\frac{\maxSupp+1}{2^{\calI-\ell}}\Big), 
    \frac{\bdown_k}{2^{\calI-i}} - \sum_{\ell=i}^{\calI-1} \Big(\frac{\maxSupp+0}{2^{\calI-\ell}}\Big)\Bigg]\\
    &= \Bigg[\frac{\bdown_k}{2^{\calI-i}} - (\maxSupp+1) \cdot \sum_{\ell=i}^{\calI-1} \Big(\frac{1}{2^{\calI-\ell}}\Big), 
    \frac{\bdown_k}{2^{\calI-i}} - \maxSupp \cdot \sum_{\ell=i}^{\calI-1} \Big( \frac{1}{2^{\calI-\ell}}\Big)\Bigg]\\
    &= \Bigg[\frac{\bdown_k}{2^{\calI-i}} - \maxSupp \cdot \sum_{\ell=i}^{\calI-1} \Big(\frac{1}{2^{\calI-\ell}}\Big) - \sum_{\ell=i}^{\calI-1} \Big(\frac{1}{2^{\calI-\ell}}\Big), \frac{\bdown_k}{2^{\calI-i}} - \maxSupp \cdot \sum_{\ell=i}^{\calI-1} \Big(\frac{1}{2^{\calI-\ell}}\Big)\Bigg].\\
\end{align*}

Since the geometric series $\sum_{\ell=i}^{\calI-1} (\nicefrac{1}{2^{\calI-\ell}})$ is less than $1,$ we know that there is at most one integer value in the interval. 
\cref{lem:m_k integer} states that $\bdown[\ii]_k$ must be integer. 
This implies that there exists exactly one integer value in the interval and
\begin{align*}
    \Bigg\lceil\frac{\bdown_k}{2^{\calI-i}} - \maxSupp \cdot \sum_{\ell=i}^{\calI-1} \Big(\frac{1}{2^{\calI-\ell}}\Big) - \sum_{\ell=i}^{\calI-1} \Big(\frac{1}{2^{\calI-\ell}}\Big)\Bigg\rceil
    =
	\Bigg\lfloor\frac{\bdown_k}{2^{\calI-i}} - \maxSupp \cdot \sum_{\ell=i}^{\calI-1} \Big(\frac{1}{2^{\calI-\ell}}\Big) \Bigg\rfloor.
\end{align*}
Thus, the number of machines $\bdown[\ii]_k$ of type $k$ in iteration $i$ can be uniquely determined by
\begin{align*}
        \bdown[\ii]_k = \begin{cases}
            0, & \text{if } i \neq \calI \text{ and } \bdown[(i+1)]_k \leq \maxSupp\\
            \Big\lfloor\frac{\bdown_k}{2^{\calI-i}} - \maxSupp \cdot \sum_{\ell=i}^{\calI-1} \big(\frac{1}{2^{\calI-\ell}}\big)\Big\rfloor, & \text{otherwise.}
        \end{cases}
    \end{align*}
\end{proof}
The algorithm to compute this value is given above, in \cref{alg:m_k^(i)}.

To ensure that the RHS can be divided by two, \cref{alg:get tilde x} describes in detail, how we reduce the solution vector in each iteration.
\begin{algorithm}
\caption{reduceSolution}
\algcomment{\textbf{\cref{alg:get tilde x}:} Determines a vector $\tilx_k^\powIteri$ such that if subtracted from $x_k^\powIteri,$ the difference $\hatx_k^\powIteri = x_k^\powIteri - \tilx_k^\powIteri$ is even in each entry.}

\textbf{Input:} Intermediate solution $x_k^\powIteri \in \ZZgeqzero^{t_k}$\\
\textbf{Output:} Solution of the sub-problem $\tilx_k^\powIteri = \big(\tilx_{k,1}^\powIteri, \dots, \tilx_{k,t_k}^\powIteri\big)^T \in \ZZgeqzero^{t_k}$

\begin{algorithmic}[1]

\State counter $\gets \maxSupp - z_k^\powIteri$ 

\For{$\ell \in[t_k]$}
    \If{$x_{k,\ell}^\powIteri$ is odd} 
    \Comment{Ensures that $\tilx_{k,\ell}^\powIteri$ and $x_{k,\ell}^\powIteri$ have the same parity}
        \State $\tilx_{k,\ell}^\powIteri \gets 1$ 
		\State counter $\gets$ counter$-1$
    \Else
        \State $\tilx_{k,\ell}^\powIteri \gets 0$
    \EndIf
\EndFor
\Statex
\While{counter $\geq 2$} 
\Comment{Ensures $\|\tilx_k^\powIteri\|_1 \leq \maxSupp$}
    \For{$\ell \in[t_k]$}
        \If{$x_{k,\ell}^\powIteri - \tilx_{k,\ell}^\powIteri \geq 2$ \textbf{and} counter $\geq 2$}
            \State $\tilx_{k,\ell}^\powIteri \gets \tilx_{k,\ell}^\powIteri + 2$
            \Comment{Add 2 to keep the parity}
			\State counter $\gets$ counter$-2$;
        \EndIf
    \EndFor
\EndWhile
\State\Return $\big(\tilx_{k,1}^\powIteri, \dots, \tilx_{k,t_k}^\powIteri\big)^T$
\end{algorithmic}
\label{alg:get tilde x}
\end{algorithm}

Next, we show that the size of the sub-problems that need to be solved in each iteration is bounded. 

\lengthJobs*
\begin{proof}
Let $\tilde{b}^\powIteri = \big(\tiln[\ii],\tilm[\ii])^T$ with $i \in [\calI]$ and $\tilm[\ii] \leq \maxSupp.$
Assume $\calA\tilx^\powIteri = \tilde{b}^\powIteri$ is feasible.

We can split the ILP of the sub-problem $\calA \tilx^\powIteri = \big(\tiln[\ii], \tilm[\ii]\big)^T$ into two parts:
\begin{align*}
    \begin{pmatrix}
        A_1 & A_2 & \dots & A_n \\
    \end{pmatrix} \cdot \tilx^\powIteri = \tiln[\ii]
\end{align*}
and
\begin{align*}
    \begin{pmatrix}
        \onevec^T &&&\\
        & \onevec^T &&\\
        && \ddots &\\
        &&& \onevec^T
    \end{pmatrix} \cdot \tilx^\powIteri = \tilm[\ii].
\end{align*}
Taking a closer look at the first part, we get
\begin{align*}
    \big\|\tiln[\ii]\big\|_\infty= \big\|\begin{pmatrix}
        A_1 & A_2 & \dots & A_n \\
    \end{pmatrix} \cdot \tilx^\powIteri\big\|_\infty= \big\|A_1 \tilx_1^\powIteri + A_2 \tilx_2^\powIteri + \cdots + A_n \tilx_n^\powIteri\big\|_\infty. 
\end{align*} 
The triangle inequality implies
\begin{align*}
    &\big\|A_1 \tilx_1^\powIteri + A_2 \tilx_2^\powIteri + \cdots + A_n \tilx_n^\powIteri\big\|_\infty
     \leq \big\|A_1 \tilx_1^\powIteri\big\|_\infty + \big\|A_2 \tilx_2^\powIteri\big\|_\infty + \cdots + \big\|A_n \tilx_n^\powIteri\big\|_\infty.
\end{align*}

By definition, we know $\big\|\tilx_k^\powIteri \big\|_1 \in \big\{0,\big\|x_k^\powIteri\big\|_1, \maxSupp-1, \maxSupp\big\}$ and thus $\big\|\tilx_k^\powIteri\big\|_1 \leq \maxSupp$ for all $k \in [n]$.
Since no value in $A_k$ is larger than $\Delta,$ it holds that $\big\|A_k \tilx_k^\powIteri\big\|_\infty \leq \maxSupp \cdot \Delta.$
Hence,
\begin{align*}
    \big\|\tiln[\ii]\big\|_\infty
    \leq \big\|A_1 \tilx_1^\powIteri\big\|_\infty + \big\|A_2 \tilx_2^\powIteri\big\|_\infty + \cdots + \big\|A_n \tilx_n^\powIteri\big\|_\infty
    \leq n \cdot \maxSupp \cdot \Delta.
\end{align*}
Now, with the definition of $D,$ it follows that $\big\|\tiln[\ii]\big\|_\infty \leq D.$
\end{proof}

We now show that in each iteration, we only need to consider job-vectors that are near the center point.
\properties*
\begin{proof}
Let $b = (\bup,\bdown)^T$ and assume $\calA x=b$ is feasible.

\subparagraph*{Property 1:} Proof by induction over the iterations.

    \textbf{Base Case:} Assume $i = \calI.$ 
    By definition it holds that $\bdown[(\calI)] = \bdown, \bup[(\calI)] = \bup$ and $\big\| \nicefrac{\bup}{2^{\calI-\calI}} - \bup[(\calI)] \big\|_\infty = 0 \leq D.$ Since $\calA x=(\bup,\bdown)^T$ is feasible, this also holds for $\calA x^\powIterCalI=(\bup[(\calI)],\bdown[(\calI)])^T.$

    \textbf{Induction Step:} Let $i \in \{3,\dots,\calI\}$ and assume that there exists a feasible point $\bup[\ii]$ near $\nicefrac{\bup}{2^{\calI-i}},$ i.e.,
            \begin{align} \label{eq:properties IV}
                \bigg\| \frac{\bup}{2^{\calI-i}} - \bup[\ii] \bigg\|_\infty \leq D
            \end{align}
            holds, and assume $\calA x^\powIteri = (\bup[\ii], \bdown[\ii])^T$ is feasible.

            We now show that there also exists a feasible point $\bup[(i-1)] \in \ZZgeqzero^r$ near $\nicefrac{\bup}{2^{\calI-(i-1)}},$ that can be determined from the reductions introduced in \cref{sec:topDown}.

            By definition of $\bup[(i-1)]$, it holds that $\bup[(i-1)] = \nicefrac{(\bup[\ii] - \tiln[\ii])}{2}.$
            It follows that
        \begin{align*}
            \bigg\|\frac{\bup}{2^{\calI-(i-1)}} - \bup[(i-1)]\bigg\|_\infty
            &= \Bigg\|\frac{\bup}{2^{\calI-i+1}} 
            -\Bigg(\frac{\bup[\ii]-\tiln[\ii]}{2}\Bigg)\Bigg\|_\infty \\
            &= \Bigg\|\frac{\bup}{2^{\calI-i+1}} 
            - \frac{\bup[\ii]}{2} 
            + \frac{\tiln[\ii]}{2}\Bigg\|_\infty \\
            &= \bigg\|\frac{1}{2} \cdot \frac{\bup}{2^{\calI-i}} 
            - \frac{1}{2} \cdot \bup[\ii] 
            + \frac{1}{2} \cdot \tiln[\ii] \bigg\|_\infty.
        \end{align*}
        The triangle inequality implies
        \begin{align*}
            \bigg\|\frac{\bup}{2^{\calI-(i-1)}} - \bup[(i-1)]\bigg\|_\infty
            &\leq \bigg\|\frac{1}{2} \cdot \frac{\bup}{2^{\calI-i}} 
            - \frac{1}{2} \cdot \bup[\ii] \bigg\|_\infty 
            + \bigg\| \frac{1}{2} \cdot \tiln[\ii] \bigg\|_\infty \\
            &= \frac{1}{2} \cdot  \bigg\|\frac{\bup}{2^{\calI-i}} 
            - \bup[\ii] \bigg\|_\infty 
            + \frac{1}{2} \cdot \bigg\| \tiln[\ii] \bigg\|_\infty.
        \end{align*}
        With assumption \eqref{eq:properties IV} and \cref{lem:jobbound}, we get 
        \begin{align*}
            \bigg\|\frac{\bup}{2^{\calI-(i-1)}} - \bup[(i-1)]\bigg\|_\infty
            &\leq \frac{1}{2} \cdot D + \frac{1}{2} \cdot D = D.
        \end{align*}
\subparagraph*{Property 2:} Let $b^\powIterTwo=(\bup[(2)],\bdown[(2)])^T$ and assume that $\calA x^\powIterTwo=b^\powIterTwo$ is feasible.
        Then it is possible to assign feasible configurations to the remaining machines.
        The definition of the iterations and their numbering imply that $\bdown[(1)]_k \leq \maxSupp$ for all machine-types $k.$
        With \eqref{eq:define tilde xi}, we get $\tilx^\powIterOne = x^\powIterOne$ and $x^{(0)} = \nicefrac{(x^\powIterOne - \tilx^\powIterOne)}{2} = \nullvec.$
        Therefore, $\bup[(0)] = \nullvec$ and $\tiln[(1)] = \bup[(1)].$ 
        Since $\calA x^\powIterTwo = b^\powIterTwo$ is feasible, there are at most $\maxSupp$ machines for each type, and with \cref{lem:jobbound}, it holds that $\big\|\bup[(1)]\big\|_\infty \leq D.$
\end{proof}

The way we calculate the $\bdown[\ii]_k$ with \autoref{alg:m_k^(i)} shows how many iterations we need for each block $k\in[n]$:
\IkFormula*
\begin{proof}
    Let $\bdown_k \in \ZZgeqzero, k \in [n]$ and $\calI \in \ZZgeqzero.$ 
    Assume that the current iteration $i$ is the one where the block $k$ vanishes, i.e., $\bdown[\ii]_k > 0$ and $\bdown[(i-1)]_k = 0.$ Thus, $i$ is the largest value with $\bdown[\ii]_k \leq \maxSupp.$
    This defines the number of non-zero iterations of that block:
    \begin{align}\label{eq:def i}
        i = \calI - \calI_k + 1.
    \end{align}
    Now, set 
    \begin{align}\label{eq:def I_k}
        I_k := \calI_k -1 =\calI-i.
    \end{align}
    We are now looking for the smallest number of iterations $I_k$ for which $\bdown[\ii]_k \leq \maxSupp$ is fulfilled.

    \cref{lem:m_k formula} and $\bdown[\ii]_k \leq \maxSupp$ imply
    \begin{align*}
        \bigg\lfloor \frac{\bdown_k}{2^{\calI-i}} - \maxSupp\cdot \sum_{\ell = i}^{\calI-1}\Big(\frac{1}{2^{\calI-\ell}}\Big)\bigg\rfloor &\leq \maxSupp \qquad \text{ and } \\
        \frac{\bdown_k}{2^{\calI-i}} - \maxSupp\cdot \sum_{\ell = i}^{\calI-1}\Big(\frac{1}{2^{\calI-\ell}}\Big) &< \maxSupp +1.
    \end{align*}
    By applying \eqref{eq:def i} and the definition of $I_k$ from \eqref{eq:def I_k}, we get
    \begin{align}\label{eq:ineq \maxSupp+1}
        \frac{\bdown_k}{2^{I_k}} - \maxSupp\cdot \sum_{\ell = \calI -I_k}^{\calI-1}\Big(\frac{1}{2^{\calI-\ell}}\Big) < \maxSupp +1.
    \end{align}
    We can rewrite the sum, since $\sum_{\ell = \calI -I_k}^{\calI-1}(\nicefrac{1}{2^{\calI-\ell}}) = \sum_{\ell = 1}^{I_k}(\nicefrac{1}{2^\ell})$ and use the closed form of the finite geometric series $\sum_{\ell = 1}^{I_k}(\nicefrac{1}{2^\ell}) = 1-\nicefrac{1}{2^{I_k}}.$
    We now derive the smallest number of iterations $I_k$ for which Inequality (\ref{eq:ineq \maxSupp+1}) holds:
    \begin{align*}
        &&\frac{\bdown_k}{2^{I_k}} - \maxSupp\cdot \sum_{\ell = \calI -I_k}^{\calI-1}\Big(\frac{1}{2^{\calI-\ell}}\Big) &< \maxSupp +1
        \\
        &\Longleftrightarrow
        &\frac{\bdown_k}{2^{I_k}} - \maxSupp\cdot \Big(1- \frac{1}{2^{I_k}}\Big) &< \maxSupp +1
        \\
        &\Longleftrightarrow
        &\frac{\bdown_k}{2^{I_k}} - \maxSupp + \frac{\maxSupp}{2^{I_k}} &< \maxSupp +1
        \\
        &\Longleftrightarrow
        &\frac{\bdown_k + \maxSupp}{2^{I_k}} &< 2\maxSupp +1
        \\
        &\Longleftrightarrow
        &\frac{\bdown_k + \maxSupp}{2\maxSupp+1} &< 2^{I_k}
        \\
        &\Longleftrightarrow
        &\log \bigg(\frac{\bdown_k + \maxSupp}{2\maxSupp+1} \bigg) &< I_k
    \end{align*}
    
    Thus, the smallest $I_k$ that fulfills Inequality (\ref{eq:ineq \maxSupp+1}) is 
    \begin{align*}
        I_k = 
        \begin{cases}
            \log \big(\nicefrac{\bdown_k + \maxSupp}{2\maxSupp+1} \big) + 1, &\txtif \log \big(\nicefrac{\bdown_k + \maxSupp}{2\maxSupp+1} \big) \text{ is integer}\\
            \big\lceil\log \big(\frac{\bdown_k + \maxSupp}{2\maxSupp+1} \big)\big\rceil, &\txtotherwise.
        \end{cases}
    \end{align*}

    Therefore, the number of iterations $\calI_k$ of machine-type $k$ with $\bdown[\ii]_k > 0$ is
    \begin{align*}
        \calI_k &= I_k+1 \\
        &= 
        \calIkFormula
    \end{align*} 
\end{proof}
This directly yields a closed form for the total number of iterations:
\IFormula*
\begin{proof}
    The total number of iterations $\calI$ is defined by $\calI = \calI_{\max} = \max_{k\in [\tau]} \calI_k.$
    \cref{lem:calI_k formula} directly implies 
    $$\calI_{\max} = \calIFormula$$
\end{proof}

We now provide a proof of our feasibility result:

\mainthm*
\begin{proof}
We first prove that the proposed algorithm correctly determines the feasibility of a given combinatorial \nfold ILP \eqref{eq:nfold}.

We prove this theorem in two steps. First, we show that the dynamic program (\cref{alg:dyn prog}) correctly determines the set $\tilde{N}^\powIteri$ of points $\tiln \in [D]^r$ for which $\calA \tilx^\powIteri = (\tiln[\ii],\tilm[\ii])^T$ is feasible. Then, we show that our algorithm returns $\trueCapital$ if and only if \eqref{eq:nfold} is feasible.

\paragraph*{Correctness of the Dynamic Program:}

In the first part, we show that the dynamic program is correct, assuming $n = 1$. 
In the latter part, we prove the property for $n \geq 2.$

\subparagraph*{Part 1:}
Assume $n = 1$ and let $\tilde{N}^\powIteri$ be the set of job-vectors that is returned by \cref{alg:dyn prog}.
Let $\nu \in [D]^r.$ As we have only one block, we have $\nu \in \tilde{N}^\powIteri$ if and only if $\calA \tilx^\powIteri = (\nu,\bdown[\ii])^T$ is feasible.

\subparagraph*{Part 2:}
Assume $n \geq 2$ and let $\tilde{N}$ be the set of job-vectors that is returned by \cref{alg:dyn prog}.
\\
$\Longrightarrow:$ Let $\nu \in [D]^r$ and assume $\nu \in \tilde{N}^\powIteri.$
Hence, it holds that $DT(\nu,n) = \trueCapital$.
We now have to show that $\calA \tilx^\powIteri = (\nu,\tilm[\ii])^T$ is feasible.

Note that $DT(\nu,n) = \trueCapital$ implies that there exists a partition $\nu_1,\dots, \nu_n$ of $\nu$ such that $\sum_{k=1}^n (\nu_j)= \nu$ and $DT(\nu_k,\{k\}) = \trueCapital$.
This means that the corresponding single-block ILPs are feasible for all $k \in [n]$.
Thus, there exist $\tilx_{\ell,k}^\powIteri \in \ZZgeqzero$ with $\sum_{\ell = 1}^{t_k} (A_{\ell, k}\cdot \tilx_{\ell,k}^\powIteri) = \nu_k$ and $\sum_{\ell=1}^{t_k} (\tilx_{\ell, k}^\powIteri) =  \tilm[\ii]_k,$  where $A_{\ell, k}$ is the $\ell$-th column of the $k$-th $A$-block.
Set 
$\tilx^\powIteri:=\big(\tilx_{1,1}^\powIteri, \dots, \tilx_{t_1,1}^\powIteri, \dots, \tilx_{1,n}^\powIteri, \dots, \tilx_{t_n,\tau}^\powIteri\big)^T$.
Then $\tilx^\powIteri$ is a solution to $\calA \tilx^\powIteri = (\nu,\tilm[\ii])^T$ which implies the desired feasibility.

$\Longleftarrow:$ Let $\calA \tilx^\powIteri = (\nu,\tilm[\ii])^T$ be feasible.
This means that there exists a solution vector $\tilx^\powIteri = (\tilx_1^\powIteri, \dots, \tilx_n^\powIteri)^T$ with $\tilx_k^\powIteri \in \ZZgeqzero^{t_k}$.
This implies that there exists a partitioning $\nu_1,\dots, \nu_n$ of $\nu$ with $\sum_{k=1}^n (\nu_k) = \nu$.
This partitioning is equivalent to $\tilx^\powIteri.$
Note that $\tilx_k^\powIteri$ solves $A_k \tilx_k^\powIteri = \nu$ with $\|\tilx_k^\powIteri\|_1 = \bdown[\ii]_k.$
Hence, $DT(\nu_k,\{k\}) = \trueCapital$ for all $k \in [n].$
The inductive step of the dynamic program (\cref{alg:dyn prog}) implies that
\begin{align*}
    DT(\nu,n) = \bigwedge_{k \in [n]} DT(\nu_k,k)= \trueCapital.
\end{align*}
Therefore, $\nu \in \tilde{N}^\powIteri.$

\paragraph*{Correctness of the Algorithm:}
Let the \nfold ILP \eqref{eq:nfold} be given. Note that while the lower RHS $\bdown[\ii]$ are unique to their iteration $i$ and pre-computed using \autoref{alg:m_k^(i)}, the upper RHS $\bup[\ii]$ are not. For a single iteration, there may be several feasible points (upper RHS) in the set $N^\powIteri$. 

To prove the correctness of our algorithm, we show that the following four properties hold:
\begin{enumerate}
    \item
    A point $\bup[(1)] \in \ZZgeqzero^r$ with $\big\|\bup[(1)] \big\|_\infty \leq D$ is in the set $N^\powIterOne$ if and only if $\calA x^\powIterOne = (\bup[(1)],\bdown[(1)])^T$ is feasible.
    \item
    Let the set $N^\powIterOne$ of feasible points $\bup[(1)] \in \ZZgeqzero^r$ with $\big\|\bup[(1)] \big\|_\infty \leq D$ be given. 
    Then a point $\bup[(2)] \in \ZZgeqzero^r$ with $\big\|\nicefrac{\bup}{2^{\calI-2}} - \bup[(2)]\big\|_\infty \leq D$ is in the set $N^\powIterTwo$ if and only if $\calA x^\powIterTwo = (\bup[(2)],\bdown[(2)])^T$ is feasible.
    \item
    Let $i \in \{3,\dots, \calI\}.$ Let the set $N^\powIteriminOne$ of feasible points $\bup[(i-1)] \in \ZZgeqzero^r$ with
    \begin{align*}
        \Big\|\frac{\bup}{2^{\calI-(i-1)}} - \bup[(i-1)] \Big\|_\infty \leq D
    \end{align*}
    be given. Then a point $\bup[\ii] \in \ZZgeqzero^r$ with $\big\|\nicefrac{\bup}{2^{\calI-i}} - \bup[\ii] \big\|_\infty \leq D$ is in the set $N^\powIteri$ if and only if $\calA x^\powIteri = (\bup[\ii],\bdown[\ii])^T$ is feasible.
    \item
    The given point $n$ is in the set $N^\powIterCalI$ if and only if $\calA x = (\bup,\bdown)^T$ is feasible.
\end{enumerate}

\subparagraph*{Property 1:} 
This follows from the correctness of the dynamic program and the fact that $\bup[(1)] = \tiln[(1)].$

\subparagraph*{Property 2:}
Let $N^\powIterOne$ be the set of feasible points $\bup[(1)] \in \ZZgeqzero^r$ with $\big\|\bup[(1)] \big\|_\infty \leq D.$
This assumption holds due to the first property of this proof.

Let $\bup[(1)] \in N^\powIterOne.$ 
Since $\bup[(1)]$ is feasible, there exists a solution vector $x^\powIterOne$ with $\calA x^\powIterOne = \tilde{b}^\powIterOne.$

$"\Longrightarrow":$ Let $\bup[(1)] \in \ZZgeqzero^r$ with $\big\|\nicefrac{\bup}{2^{(\calI-2)}} - \bup[(2)] \big\|_\infty \leq D$ and assume $\bup[(2)] \in N^\powIterTwo.$
This implies that there exists $\bup[(1)] \in N^\powIterOne$ and an $\tiln[(2)] \in \tilde{N}^\powIterTwo$ such that $\bup[(2)] = 2\bup[(1)] + \tiln[(2)].$
The correctness of the dynamic program and the existence of $\tiln[(2)]$ imply that $\calA \tilx^\powIterTwo = \tilde{b}^\powIterTwo$ is feasible.
Set $x^\powIterTwo := 2x^\powIterOne + \tilx^\powIterTwo$ and $b^\powIterTwo := 2b^\powIterOne + \tilde{b}^\powIterTwo.$ 
It follows that $\calA x^\powIterTwo = b^\powIterTwo$ and hence, which implies the desired feasibility.

$"\Longleftarrow":$ Assume $\calA x^\powIterTwo = (\bup[(2)],\bdown[(2)])^T$ with $\|\bup[(2)]\|_\infty \le D$ holds.
Property 2 of \cref{lem:properties norm n leq D} now states that there exists an $x^\powIterOne$ with $\calA x^\powIterOne = \big(\bup[(1)], \bdown[(1)]\big)^T$ 
and $\big\| \bup[(1)] \big\|_\infty \leq D.$

Set $\hat{N}^\powIterTwo := \big\{\hatn[(2)]\in \ZZgeqzero^r \mid \hatn[(2)] = 2\bup[(1)], \bup[(1)] \in N^\powIterOne\big\}.$
With the correctness of the dynamic program, we know that $\tilde{N}^\powIterTwo$ is the set of all feasible points with $\big\|\tiln[(2)]\big\|_\infty \leq D$ and $\tilm[(2)] = \bdown[(2)] - 2\bdown[(1)].$
We now combine all $\hatn[(2)] \in \hat{N}^\powIterTwo$ with all $\tiln[(2)] \in \tilde{N}^\powIterTwo$ and only keep the $\bup[(2)] \in \ZZgeqzero^r$ with $\big\|\nicefrac{\bup}{2^{\calI-2}} - \bup[(2)]\big\|_\infty \leq D$ in the set 
\begin{align*}
    N^\powIterTwo = \big\{\bup[(2)] \in \ZZgeqzero^r \mid \bup[(2)] = \hatn[(2)] + \tiln[(2)], \hatn[(2)] \in \hat{N}^\powIterTwo, \tiln[(2)] \in \tilde{N}^\powIterTwo,\\
    \big\|\nicefrac{\bup}{2^{\calI-2}} - \bup[(2)]\big\|_\infty \leq D \big\}.
\end{align*}
This is the set of all feasible points near $\nicefrac{\bup}{2^{\calI-2}}.$ 
Since $\calA x^\powIterTwo = (\bup[(2)],\bdown[(2)])^T$ is feasible, it follows that $\bup[(2)] \in N^\powIterTwo.$

\subparagraph*{Property 3:}
Let $i \in \{3,\dots, \calI\}$ and assume the set $N^\powIteriminOne$ of feasible points $\bup[(i-1)] \in \ZZgeqzero^r$ with 
\begin{align*}
    \bigg\|\frac{\bup}{2^{\calI-(i-1)}} - \bup[(i-1)]\bigg\|_\infty \leq D
\end{align*}
is given.
Note that there exists a solution vector $x^\powIteriminOne$ with 
$\calA x^\powIteriminOne = b^\powIteriminOne$ for each $\bup[(i-1)] \in N^\powIteriminOne$.

$"\Longrightarrow":$ Let $\bup[\ii] \in \ZZgeqzero^r$ with $\big\|\nicefrac{\bup}{2^{(\calI-i)}} - \bup[\ii]\big\|_\infty \leq D$ and assume $\bup[\ii] \in N^\powIteri.$
This means that there exist some $\bup[(i-1)] \in N^\powIteriminOne$ and $\tiln[\ii] \in \tilde{N}^\powIteri$ such that $\bup[\ii] = 2\bup[(i-1)] + \tiln[\ii].$
The correctness of the dynamic program and the existence of $\tiln[\ii]i$ imply that there exists $\tilx^\powIteri$ with $\calA \tilx[\ii] = \tilde{b}^\powIteri.$

Set $x^\powIteri := 2x^\powIteriminOne + \tilx^\powIteri$ and $b^\powIteri := 2b^\powIteriminOne + \tilde{b}^\powIteri.$ 
It follows that $\calA x^\powIteri = b^\powIteri$ which implies the desired feasibility.

$"\Longleftarrow":$ Assume there exists $x^\powIteri \in \ZZgeqzero^r$ such that $\calA x^\powIteri = b^\powIteri$ with $\|\bup[\ii]\|_\infty \le D$ holds.
Property 1 of \cref{lem:properties norm n leq D} now implies that there exists $x^\powIteriminOne$ with $\calA x^\powIteriminOne = b^\powIteriminOne$ and 
\begin{align*}
    \bigg\|\frac{\bup}{2^{\calI-(i-1)}} - \bup[(i-1)]\bigg\|_\infty \leq D.
\end{align*}
Set $\hat{N}^\powIteri := \big\{\hatn[\ii] \in \ZZgeqzero^r \mid \hatn[\ii] = 2\bup[(i-1)], \bup[(i-1)] \in N^\powIteriminOne\big\}.$
With the correctness of the dynamic program, we know that $\tilde{N}^\powIteri$ is the set of all feasible points with $\big\|\tiln[\ii]\big\|_\infty \leq D$ and $\tilm[\ii] = \bdown[\ii] - 2\bdown[(i-1)].$
We now combine all $\hatn[\ii] \in \hat{N}^\powIteri$ with all $\tiln[\ii] \in \tilde{N}^\powIteri$ and only keep the $\bup[\ii] \in \ZZgeqzero^r$ with $\big\|\nicefrac{\bup}{2^{\calI-i}} - \bup[\ii]\big\|_\infty \leq D$ in the set 
\begin{align*}
    N^\powIteri = \big\{\bup[\ii] \in \ZZgeqzero^r \mid \bup[\ii] = \hatn[\ii] + \tiln[\ii], \hatn[\ii] \in \hat{N}^\powIteri, \tiln[\ii] \in \tilde{N}^\powIteri, \\
    \big\|\nicefrac{\bup}{2^{\calI-i}} - \bup[\ii]\big\|_\infty \leq D\big\}.
\end{align*}
This is the set of all feasible points near $\nicefrac{\bup}{2^{\calI-i}}.$ 
Since $\calA x^\powIteri = b^\powIteri$ is feasible, it follows that $\bup[\ii] \in N^\powIteri.$
    
\subparagraph*{Property 4:} 
Set $\bup[(\calI)] := \bup$ and $\bdown[(\calI)] := \bdown.$
This implies 
\begin{align*}
    \Big\|\frac{\bup}{2^{(\calI-\calI)}} - \bup[(\calI)]\Big\|_\infty = \|\bup - \bup\|_\infty= 0\leq D.
\end{align*}

The third property of this proof states that $\bup[\ii] \in \ZZgeqzero^r$ with $\big\|\nicefrac{\bup}{2^{(\calI-i)}} - \bup[\ii]\big\|_\infty \leq D$ is in the set $N^\powIteri$ if and only if $\calA x^\powIteri = b^\powIteri$ is feasible.
With $\big\|\nicefrac{\bup}{2^{(\calI-\calI)}} - \bup[(\calI)]\big\|_\infty \leq D$, we know that this holds also for $\bup[(\calI)].$ 
Thus, $\bup \in N^\powIterCalI$  if and only if $\calA x = b$ is feasible.

\paragraph*{Running Time}
The running time of the preprocessing is dominated by the dynamic program. 
In the worst case, it is called $\Oh(n)$ times.
Creating the base table (BT) of the DP, we have to solve $n (K\Delta+1)^r$ small sub-problems, each having a running time of \eqref{eq:pcmax rt}.
The dynamic program then constructs a dynamic table (DT) with $n(D+1)^r$ entries. Each can be calculated by doing a boolean convolution using FFT which takes time $\Oh(n(D+1)^r \cdot \log(n(D+1)^r))$ (see e.g.~\cite{BN21}).
With $t_k \le h$ for all $k\in [n]$, the total running time for the preprocessing amounts to
\begin{align*}
    &\underbrace{n (K\Delta+1)^r \cdot O(\sqrt{r+1} \Delta)^{(1+o(1))(r+1)} + O(h(r+1))}_\text{Base Case} + \underbrace{n^2 (D+1)^\Or \cdot \log(n (D+1)^r)}_\text{Inductive Step}\\
    &= (D+1)^\Or
\end{align*}

The number of iterations, we need to combine the solutions is generally bounded by $\log(\bdown_{\max})$. However, note that if the upper RHS is relatively small in all entries, it becomes $\nullvec$ before $\bdown$ does. In order to be feasible in that case, we only need to check whether a $\nullvec$ in each $A_i$ where the corresponding entry in the lower RHS is not 0, exists. Thus, these sub-problems are trivial to solve and therefore, we may say that we need $\log(b_{\text{def}})$ iterations, where $b_{\text{def}} := \min(\bup_{\max},\bdown_{\max})$. Each set $N^\powIteri, \tilde{N}^\powIteri$ and $\hat{N}^\powIteri$ contains at most $(D+1)^r$ vectors. Therefore, determining each set is possible in time $D^\Or.$
With $h \le (\Delta+1)^r$ this yields to a total running time of $\RTfeasibilityO$.
\end{proof}

\section{Applications \-- Omitted details}\label{sec:appendix-sched}
In this appendix, we elaborate how our algorithm can be applied to scheduling on uniform machines, with respect to makespan minimization and santa claus, the \closeststring problem and the \imbalance problem.

\subsection{Scheduling on Uniform Machines}\label{sec:scheduling}

In this section, we show how to use our combinatorial \nfold approach to solve the problem of scheduling jobs on uniform machines with the objective of minimizing the maximum completion time or maximizing the minimum completion time, i.e., $\QCmax$ and $\QCmin$ respectively. We begin by summarizing a novel approach by Rohwedder to solve $\QCmax$ given in~\cite{rohwedder2025}.

\paragraph*{Makespan Minimization}
\label{parac:QCmax}
 His core idea is simple. He shows that you can separate the scheduling problem into two sets of machines, small ($S$) and big ($B$). For a given makespan guess $T$ and a machine-type $k$, we call $T_k=\lfloor s_kT\rfloor$ the guessed load of that type. With a slight abuse of notation, we utilize $T_k$ as the guessed load of a machine $m^{(k)}$ as well.
Small machines have a small guessed load, i.e., $T_k\leq \pmax^4\forall m^{(k)}\in S.$ This directly yields small configuration sizes for small machines, and, relevant for our approach to solve the corresponding ILP, a largest entry of $\Delta\leq \pmaxOne.$ The elegance lies in the way Rohwedder approaches big machines. He shows that there always exists a pivot size $a\in \{p_1,\dots,p_d\}$, such that jobs of this size are present in large quantities on big machines. 
For that pivot size, he shows that an optimal schedule can be generated in two phases. First, we place a small configuration on each machine. Machines are assigned configurations based on the parity of their load. A machine with guessed load $T_k$ will only get assigned configurations $c$ with congruent size, i.e., $T_k\mod a \equiv p^Tc \mod a.$ In phase two, after each machine is assigned a configuration, we can place the remaining jobs in a greedy manner to complete the schedule. For this procedure to work, it is required that the jobs not placed in phase one are either of size $a$ or, for other sizes, are present in multiples of $a$. Because configurations of congruent sizes are treated alike in phase one, we only need to consider configurations that place up to $(a-1)$ copies of a job. Furthermore, as machine-types of congruent sizes are treated alike in phase one, we only have $a$ different types of big machines.
We refer the reader to~\cite{rohwedder2025} for the whole procedure, but go into more detail when discussing our approach of solving $\QCmin.$ As a direct consequence, the configurations for both big and small machines have their entries bounded by $\pmaxOne.$ Also, the number of distinct machine-types $\tau$ is bounded by $\pmax^4+a$ as we consider at most $\pmax^4$ types of small machines and $a \le \pmax$ types of big machines. Thus, we can use the algorithm given in \cref{sec:algorithm} to solve the configuration ILP in time $\RTsched\leq \RTqcmaxO$ as we have $\log(m_{\text{def}})\le |I|$ where $m_{\text{def}}:=\min(n_{\max},m_{\max})$. Then, extending this initial solution to admit all jobs can be done with a greedy procedure detailed in~\cite[Section 2.1]{rohwedder2025}, and thus requires only time polynomial in $N$ to complete. The guessed makespan can be computed via a simple binary search framework, yielding only logarithmic overhead. For details on that procedure, we refer to the next section. In total, we have the following corollary.
\begin{corollary}
    $\QCmax$ can be solved in $(\RTsched + O(N))d\log(\tau_I N\pmax) = \RTqcmaxO$.
\end{corollary}
We note that Rohwedder achieves a running time of $(d\pmax)^{O(d)}(h+M)^{O(1)}$, where $h$ is the total number of columns in the configuration ILP, to decide each guessed makespan in a similar binary search framework to ours.
Next, we show how to extend this idea to also solve $\QCmin.$ In the process, we detail some steps given by Rohwedder in~\cite{rohwedder2025} for $\QCmax.$

\paragraph*{Santa Claus}
In this section, we show how to adapt the techniques given in~\cite{rohwedder2025} combined with our combinatorial \nfold algorithm to solve $\QCmin$ in time $\RTqcmaxO$.

\noindent\textbf{Bounding the Configurations }
As a first step, we conduct a binary search on the optimal minimum completion time $\Cmin$. Given a guess $T\in \mathbb{Q}_{\ge 0}$ we determine whether there is an assignment of jobs onto machines such that for the load of each machine $\sum_{j:\sigma(j)=k}p_j=L_k\ge T_k= \lceil s_kT \rceil$ holds. The optimal value has the form $\nicefrac{K}{s_k}$ for some $K\in [N\pmax]_0$ and some $k\in [\tau_I].$ Thus, a binary search takes $O(\log(\tau_I N\pmax))$ time, which is polynomial in the input. For any given guess $T$, we then compute the guessed loads $T_i=\lceil s_iT\rceil$ each machine-type must at least cover. Then, we create $\ell$ dummy jobs of size $-1$ such that we can exactly cover each machine, i.e., $\ell=\sum_{j\in [d]}n_jp_j-\sum_{k\in [\tau_I]} m_kT_k.$ We require this exact amount of dummy jobs to augment the configurations in the final steps of the algorithm. By adding the dummy jobs, we also add a new job-type $d+1$ with multiplicity $n_{d+1} = \ell$ and processing time $p_{d+1} = -1$. 
As a next step, we generate the required configurations to solve this scheduling instance as a combinatorial \nfold ILP. For this, the following observation is required. This can be seen by an exchange argument (\cref{sec:appendix-sched}).
\begin{restatable}{lemma}{cminload}\label{lem:Cminload}
    Let $I$ be an instance of $\QCmin$ and $OPT(I)$ be the optimal minimum completion time. Then, each machine-type $k\in[\tau_I]$ has a load of $\lceil OPT(I)s_k\rceil\leq L_k\leq OPT(I)s_k+\pmax$. 
\end{restatable}
\begin{proof}
    The proof follows by contradiction. Assume there is an optimal solution $OPT$ such that there is a machine $m^{(k)}$ with load $L_k>OPT(I)\cdot s_k+\pmax.$ Then, we can remove any one job placed on $m^{(k)}$ and place it on the machine with the smallest completion time. As we only remove a single job, the load of $m^{(k)}$ does not fall below $OPT(I)$. Furthermore, the machine on which the removed job is placed only has its load increased, so $OPT'(I)\ge OPT(I)$ holds. This procedure can then be iterated for all machines with a load that is too large. 
\end{proof}

 We split the machines into two sets, small ($S$) and big ($B$). A machine $m^{(k)}$ is small if $T_k\leq \pmax^4$, i.e., its guessed load is small. Because of \cref{lem:Cminload}, we know that each (small) machine has at most $\pmax$ many dummy jobs assigned to it. The remaining machines are considered big. We construct the configurations depending on whether the corresponding machines are small or big.

For small machines, since their load is small, each entry of the configuration is bounded by $\pmax^4$, except for the dummy jobs, which are bounded by $\pmax.$ Thus, a small machine $m^{(k)} \in S$ has the set of configurations $\mathcal{C}_k=\{ c\in \mathbb{Z}^{d+1}_{\geq 0}| p^tc=T_k, c_i\leq \pmax^4\forall i\in [d], c_{d+1}\leq \pmax\}.$ Let $\mathcal{C}_S=\{\mathcal{C}_k|m^{(k)}\in S\}$ be the set of all small configurations.
For the big machines, we first guess a pivot size $a \in \{p_1,\dots,\pmax\}$. This pivot size is chosen such that jobs of that size are very common on big machines, i.e., there are at least $\pmax^2|B|$ many such jobs present. The following Lemma shows that such a size always exists.

\begin{restatable}[\cite{rohwedder2025}]{lemma}{ajobs}\label{lem:a-jobs}
    Let $I$ be an instance of $\QCmin$ and $\sigma^*$ be an optimal schedule for $I$ and let $B$ be the set of machines with load larger than $\pmax^4$ in $\sigma^*$. Then there exists a job-type $j$ such that at least $\pmax^2|B|$ jobs of size $p_j$ are scheduled on those machines.
\end{restatable}
\begin{proof}
    The argument is a simple pigeonhole: The load of big machines is at least $\pmax^4,$ so there must be at least $\pmax^3$ jobs on those machines. In total there are at least $\pmax^3 |B|$ jobs on big machines. As there are at most $d\le \pmax$ job sizes, one such size must be present at least $\pmax^2 |B|$ times.
\end{proof}

Using this pivot size $a$ we now compute the parity of each machine size $k$, i.e., $T_k \mod a.$ For each parity value, we construct only configurations with a total load of that parity. Furthermore, we know that each job size is only ever present $(a-1)$ times, as there would be a configuration of equivalent parity otherwise. Finally, the number of jobs $j$ with $p_j=a$ is $0$ for all configurations by the same argument. Thus, the set of configurations of a parity-type $k\in [a-1]_0$ is $\mathcal{C}_k=\{c\in \ZZgeqzero^{d+1}\mid p^Tc\mod a =k, c_{\max}<a\}.$ Let $\mathcal{C}_B=\{\mathcal{C}_k|m^{(k)}\in B\}$ be the set of configurations for big machines.
Note that these configurations do not fill up the big machines entirely. However, due of the construction of our dummy jobs, we ensure that the jobs that remain unplaced by these configurations have a total size that is a multiple of $a.$ This can be seen because small machines are filled exactly, and big machines are filled such that a load that is a multiple of $a$ is missing on each machine. As we have $\sum_{j\in [d]}n_jp_j-\ell=\sum_{k\in [\tau_I]} m_kT_k,$ i.e., the total size of jobs and loads is identical with the addition of dummy jobs, the total size of unscheduled jobs is a multiple of $a$. Note that we need to take care when constructing our ILP that the remaining unscheduled jobs can be placed in a greedy manner. We achieve this through the construction of our slack variables.

\paragraph*{Building and Solving the Configuration ILP}
Knowing that the size of feasible configurations is bounded, we build the \textit{configuration ILP} as in three steps:
\begin{enumerate}
    \item For small machines, the entries in $\mathcal{C}_S$ form the blocks of configurations. More precisely, let $\tau_S := |\mathcal{C}_S|$ be the number of small machine-types, then define matrices $A_i$ such that the vectors in the $i$-th set of $\mathcal{C}_S$ are the columns in $A_i$ with $i\in[\tau_S]$.
    \item For big machines, the entries in $\mathcal{C}_B$ form the blocks of configurations. Set $\tau_B := |\mathcal{C}_B| \le a$. Then we define the matrices $A_{\tau_S + i}$ analogously to 1.\ such that the vectors in the $i$-th set of $\mathcal{C}_B$ are the columns in $A_{\tau_S + i}$ with $i\in[\tau_B]$. Let $\tau:=\tau_S+\tau_B$ be the number of blocks.
    \item The final block is a slack block which represents the yet unscheduled jobs, i.e., the removed jobs of size $a$ and the bundles of $a$ jobs of a single type. To ensure, that only jobs of size $a$ or bundles of $a$ jobs are unassigned, we define 
        \begin{align*}
            A_{\tau+1} := \begin{pmatrix}
            \alpha_1 &0 & \dots & 0 & 0\\
            0       & \alpha_2 & & \vdots&\vdots\\
            \vdots &&\ddots & 0 &0\\
            0&\dots&0&\alpha_{d+1} & 0
        \end{pmatrix} 
        \qquad\text{where }
            \alpha_j := 
            \begin{cases}
                1, & \txtif p_j = a\\
                a, & \txtother.
            \end{cases}
        \end{align*}
        We set the corresponding local constraint to $\onevec^T x_{\tau+1} = N$. Note that the $\nullvec$-column in $A_{\tau+1}$ adds a slack variable to the local constraint. Therefore, we can keep the equation while ensuring that at most $N$ jobs or bundles are unassigned. 
\end{enumerate}

\cref{lem:a-jobs} states that in an optimal schedule, there are at least $\pmax^2 \cdot |B|$ jobs of size $a$ on the big machines. Since the configurations $\mathcal{C}_S$ and $\mathcal{C}_B$ only schedule jobs of size $a$ on small machines, we define the job-vector $n'$ by
$$n_j':= \begin{cases}
                n_j - \pmax^2 \cdot |B|, & \txtif p_j = a\\
                n_j, & \txtother.
            \end{cases}$$

This construction results in the following configuration ILP: 
\begin{equation} \label{eq:configILP}
\begin{array}{rl}
    \calA x := \begin{pmatrix}
        A_1 & A_2 & \dots & A_{\tau+1} \\
        \onevec^T &&&\\
        & \onevec^T &&\\
        && \ddots &\\
        &&& \onevec^T
    \end{pmatrix} \cdot \begin{pmatrix}
        x_1 \\ x_2 \\ \vdots \\x_{\tau+1}
    \end{pmatrix} &= \begin{pmatrix}
        n'\\
        m\\
        \calN
    \end{pmatrix}\\
    x_k &\in \ZZgeqzero^{|\calC_k|} \quad \forall k\in[\tau]\,\quad
    x_{\tau+1} \in \ZZgeqzero^{d+1}
\end{array}
\end{equation}
The feasibility of the configuration ILP means: There exists an assignment such that the load of the small machines matches $T_i$ (including dummy jobs), the parity of the big machines matches $T_i \mod a$ and for each job-type, the load of the unassigned jobs is divisible by $a.$

Note that \eqref{eq:configILP} is a combinatorial \nfold ILP. Thus, we can apply our algorithm proposed in \cref{sec:algorithm} to solve it.
Regarding the running time, we plug in the following estimations. 
Since there are at most $\pmax^4$ types of small machines and at most $a$ types of big machines, we have $\pmax^4+a+1 = \pmaxOne$ blocks and each block contains at most $\pmax^{O(d)}$ columns. The largest entry in $\calA$ is the largest entry in the configurations, which is bounded by $\pmax^4$.
Since the largest entry on the RHS is $N$, \cref{thm:opt} implies that we can solve \eqref{eq:configILP} in time $(\tau d \pmax)^\Od \log(N) \le \RTqcmaxO$.

\paragraph*{Assigning Remaining Jobs/Bundles }
As a final step, we have to add the unassigned jobs of size $a$ and bundles of $a$ jobs in polynomial time.
\begin{restatable}{lemma}{bundles}
    A solution of \eqref{eq:configILP} can be greedily augmented to a solution of $\QCmin.$
\end{restatable}
\begin{proof}
    We first outline the procedure we utilize and then show its correctness.
    \begin{enumerate}
    \item Assign all yet unassigned bundles of dummy jobs to an arbitrary big machine.
    \item Assign all yet unassigned other bundles of jobs of same size greedily. 
    \item Assign all yet unassigned jobs of size $a$ greedily.
\end{enumerate} 
We know that a solution to \eqref{eq:configILP} exactly fills small machines. Furthermore, big machines are filled such that a multiple of $a$ is missing from their guessed load. As we also have created $\ell$ dummy jobs to exactly cover each machine, i.e., $\ell= \sum_{j\in [d]}n_jp_j-\sum_{k\in [\tau]} m_kT_k,$ the total size of unplaced jobs is a multiple $a.$ By the construction of our slack-block $A_{\tau+1}$ jobs of a size $\neq a$ are only left unplaced in multiples of $a.$ Thus, we can place all dummy jobs onto an arbitrary big machine. Afterwards, $ \sum_{j\in [d]}n_jp_j-\ell=\sum_{k\in [\tau]} m_kT_k$ still holds. 

Next, we place the bundles of jobs of size $\neq a.$ These are placed in bundles of $a$ jobs each, onto machines $i$ such that adding the bundle does not exceed $T_i.$ Let $z_{j}$ be the current assignment of jobs at some point during this step. Then there is a big machine $m^{(i)}\in B$ with $p^Tz_j\leq T_i-\pmax^2$. This is because \eqref{eq:configILP} does not place $\pmax^2|B|$ jobs of size $a$. Therefore, the maximum load of jobs unplaced in \eqref{eq:configILP} is $\sum_{i\in S}T_i+ \sum_{k\in B}T_k-\pmax^2 |B|.$  Thus, by pigeonhole principle, there must be a big machine $i$ with load $p^Tz_j\leq T_i-\pmax^2$ at any step of this assignment. As each bundle contains $a\leq \pmax$ jobs of size $\leq \pmax$, adding these to this machine does not exceed $T_i.$ After adding all unplaced bundles, the remaining space on each machine is a multiple of $a$, as each bundle contained exactly $a$ jobs.

In the last step we add single jobs of size $a$. As the remaining space on all machines is a multiple of $a$ and $ \sum_{j\in [d]}n_jp_j-\ell=\sum_{k\in [\tau]} m_kT_k,$ i.e., the remaining space is exactly the size of unscheduled jobs, we can place these greedily to exactly fill each machine. 
\end{proof}

Note that adding a bundle of $a$ dummy jobs clears some space on a machine because our dummy jobs have negative size. Therefore, temporarily, a machine might have negative load. However, this will be equalized by the other assignments.

After assigning all jobs, we have constructed a schedule for the guessed value of $C_{\min},$ $T$.
\begin{theorem}
    $\QCmin$ can be solved in $\RTqcmaxO$.
\end{theorem}
To complete the discussion of applying our techniques to the world of scheduling, we note that the algorithm presented matches the theoretical lower bound of $\pmax^{O(d^{1-\delta})}$ given in~\cite{JansenKZ25}. For the complete discussion, we defer the reader to the appendix.

 Further applications and details on \cref{cor:closest string,cor:imbalance} are provided in \cref{sec:closest string,sec:imbalance}.

\paragraph*{ETH Lower Bounds}
To complete the discussion of applying our techniques to the world of scheduling, we note that the algorithm presented matches the theoretical lower bound of $\pmax^{O(d^{1-\delta})}$ given in~\cite[Theorem 2]{JansenKZ25}. This bound is derived by the Exponential Time Hypothesis. The exponential time hypothesis, formulated by Impagliazzo and Paturi~\cite{ImpagliazzoP01}, is a computational hardness assumption. It states that the NP-complete problem 3-SAT cannot be solved in time $2^{o(\ell)}\text{poly}(|I|)$, where $\ell$ is the number of variables in the 3-SAT formulation. Under this assumption, we can further determine similar running time lower bounds for problems that 3-SAT can be reduced to. One such problem is scheduling on identical machines, i.e., $P||\Cmax.$ For this problem, Chen, Jansen and Zhang~\cite[Theorem 2]{ChenJZ18} show that, for any $\delta>0,$ there is no $2^{O(n^{1-\delta})}$ time algorithm solving $P||\Cmax.$ We note that their reduction fills each machine to exactly the optimal makespan. Thus, if we can not solve $P||\Cmax$ optimally, we can also not solve the corresponding $P||\Cmin$ instance with the same optimal value. This is because for some machine to have a completion time larger than the optimal value, there must also be one machine with a completion time lower than the optimal value. Jansen, Kahler and Zwanger~\cite{JansenKZ25} further analyzed the parameters involved in the reduction. Through this, they show that this yields a lower bound of $\pmax^{O(d^{1-\delta})}$ for $P||\Cmax,$ and by extension also for $P||\Cmin.$ As $\QCmax$ and $\QCmin$ are generalizations of $P||\Cmax$ and $P||\Cmin$ respectively, these hardness results extend to the uniform machine setting. Thus, an algorithm solving $\QCmax$ or $\QCmin$ in time $\pmax^{O(d^{1-\delta})}$ contradicts the ETH.

\subsection{Closest String Problem}\label{sec:closest string}
We define column-types as follows: Consider two column-vectors $c$ and $c'$ of length $k$. Now, $c$ and $c'$ are of the same column-type if for all $1 \le i < j \le k$, we have $c_i = c_j$ if and only if $c_i' = c_j'$. Let $T \in \ZZ_{\ge1}$ be the total number of column types.
Then we can use the following ILP formulation by Gramm \etal~\cite{GNR03}:
\begin{equation*}
\begin{array}{rcll}
    \sum_{e \in [k]} \sum_{f \in T} d_H(e,f_j) x_{f,e} &\le& d &\forall j \in [k]\\
    \sum_{e\in [k]} x_{f,e} &=& b^f &\forall f \in [T]\\
    x_{f,e} &\in& \ZZgeqzero &\forall (f,e) \in [T] \times [k]
\end{array}
\end{equation*}
Note that after this is a combinatorial \nfold ILP \eqref{eq:nfold} with $T$ blocks and inequalities in the global constraints. The blocks in the global constraints have $k$ rows and $k$ columns and since all entries in the matrix are binary it holds that $\Delta \le 1$.
To apply the algorithm presented in \cref{sec:algorithm}, we need to transform the inequalities into equations. We can do that by adding a slack block $A_{T+1} := \begin{pmatrix}
    I_k & \nullvec
\end{pmatrix}.$ The $\nullvec$ adds a slack variable in the corresponding local constraint, where we set the RHS to $dk$.
Now, applying our algorithm yields a running time of $((T+1) k)^{O(k)} \log(\min(d,\max(\max_{f\in[T]}b^f,dk))) \le ((T+1) k)^{O(k)} \log(L)$ which results in the following corollary.
\closestStringCor*
In the worst case this matches the currently best result by Knop et al.~\cite{KnopKM20} with $k^{O(k^2)} O(\log(L))$. However, in the setting, when the number of column-types $T$ is bounded by $k^{O(1)}$, we are able to reduce the quadratic dependency on $k$ in the exponent to a linear one. Also, in many other application settings, when the bounds on the number of column-types are for instance $k^{\log k}$ or $k^{k^{0.5}}$, we achieve an improvement in the running time.

\subsection{Imbalance Problem}\label{sec:imbalance}
We now parameterize \imbalance by the size $k$ of a vertex cover (VC) $C = \{c_1,\dots, c_k\}$. Observe that the VC can be determined in FPT time~\cite{CyganFKLMPPS15}.
We use the approach of Fellows \etal~\cite{FellowsLMRS08} to formulate this problem as an ILP. After that, we explain how this ILP can be slightly modified such that it has the combinatorial \nfold structure \eqref{eq:nfold} and apply our algorithm.

First, we loop over all $k!$ orderings $\pi_C$ of the VC. W.l.o.g.\ we can assume $\pi_C(c_1) < \pi_C(c_2) < \dots < \pi_C(c_k)$. Since $C$ is a VC, the remaining vertices $I := V\setminus C$ form an independent set. Now, we aim to optimally insert the vertices of $I$ into the ordering of $C$. Note that there are $k+1$ slots where a vertex $v\in I$ might be inserted: either to the left of all $c\in C$ which we call slot 0, i.e., $\pi(v)<\pi(c_1)$; in between two vertices $c_i,c_{i+1}\in C$ which we call slot $i\in[k-1]$, i.e., $\pi(c_i)<\pi(v)<\pi(c_{i+1})$; or to the right of all $c\in C$ which we call slot $k$, i.e., $\pi(c_k)<\pi(v)$.
Now, define the \textit{type} of a vertex $v \in I$ as its neighborhood $N(v)\subseteq C$ and let $T$ be the number of occurring types. Observe that we may have up to $2^k$ many vertex-types. 

As $I$ is an independent set, the inner order of those vertices at a certain slot does not matter. Therefore, the main problem is now to decide how many vertices of each type have to be inserted at which slot to achieve optimality w.r.t.\ $\pi_C$. Define $e_i := |N(c_i) \cap \{c_1,\dots,c_{i-1}\}|-|N(c_i) \cap \{c_{i+1},\dots,c_k\}|$ and let $|I(S)|$ be the number of vertices of type $S$.
Now we can formulate this as the following ILP where $x_S^i$ is number of vertices of type $S$ at slot $i$, $z_S^i$ is the imbalance of a vertex of type $S$ if placed at slot $i$ and $y_i$ is the total imbalance of $c_i$.
\begin{align}
    &\min \sum_{i=1}^k y_i + \sum_{i=0}^k\sum_{S\in[T]} z_S^i s_S^i &\\
    \text{s.t.} & -y_i + \sum_{\substack{S\in[T] \\ c_i\in S}} \Big(\sum_{j=0}^{i-1}x_S^j - \sum_{j=i}^k x_S^j\Big) \le -e_i &\forall i\in[k]\\
    & -y_i - \sum_{\substack{S\in[T] \\ c_i\in S}} \Big(\sum_{j=0}^{i-1}x_S^j - \sum_{j=i}^k x_S^j\Big) \le e_i &\forall i\in[k]\\
    & \sum_{i=0}^k x_S^i = |I(S)| &\forall S\in[T]\\
    &x_S^i \in \ZZgeqzero &\forall i\in[k]_0, S\in[T]\\
    &y_i \in \ZZgeqzero &\forall i\in[k]
\end{align}
For more details for the construction we refer to~\cite{FellowsLMRS08} and~\cite{CyganFKLMPPS15}.

This ILP can be transformed into a combinatorial \nfold \eqref{eq:nfold}. First, we add a local constraint for the variables $y_i$. We know that each $y_i$ is upper bounded by $n$, thus the sum is bounded by $k(n-1)$. We add a slack variable $y_{k+1}$ and add the constraint $\onevec_{k+1}^T y = k(n-1)$.
Similar to the slack block in the \nfold ILP of the \closeststring problem, we also add a slack block (and a corresponding local constraint) in this ILP to turn the $2k$ inequalities into equations. The RHS of the additional local constraint is $2k(n-1)$.
In total, this results in the combinatorial \nfold ILP \eqref{eq:nfold} with $(T+1)(k+1)+2k+1$ variables and $T+2$ blocks, each (except for the slack block) having $k+1$ columns and $2k$ rows. The slack block has $2k+1$ columns and $2k$ rows. Note that the largest absolute value $\Delta$ in the constraints is $1$.
Now, applying the algorithm proposed in \cref{sec:reduction}, we can solve this \nfold in time $((T+2)k)^{O(k)} O(\log(kn))$. Adding the $k! = O(k^k)$-overhead for guessing the ordering of the VC, we achieve the following result.
\imbalanceProb*
In the worst case, the number of vertex-types is $2^k,$ which results in a parameter dependency of $2^{O(k^2)}$. Applying the algorithm by Knop \etal~\cite{KnopKM20} to this problem achieves a running time of $k^{O(k^2)}O(T^3 \langle I \rangle) + \mathcal{R}$, where $\langle I \rangle$ is the input size and $\mathcal{R}$ is the time required for one call to an optimization oracle (see \cite[Theorem 3]{KnopKM20} for details). Therefore, we achieve an improvement in the parameter dependency when solving \imbalance parameterized by VC. Also note that the dependency in the exponent is reduced to linear when the $T$ is bounded by $k^{O(1)}$. Other application settings, bound $T$ by $k^{\log k}$ or $k^{k^{0.5}}$ in these cases, we also achieve an improvement in the running time.

\end{document}